\documentclass[12pt]{article}
\usepackage{amsmath}
\usepackage{graphicx,psfrag,epsf}
\usepackage{enumerate}
\usepackage{natbib}
\usepackage{url} 
\usepackage{amssymb}
\usepackage{bm}
\usepackage{bigints}
\usepackage{parskip}
\usepackage{float}
\usepackage{amsthm}
\usepackage{graphicx}
\usepackage{subcaption}
\usepackage[colorlinks, citecolor = blue, linkcolor = blue]{hyperref}
\usepackage{algorithm2e}

\newcommand{\blind}{1}
\def\ind{\text{ind}}
\def\Var{\text{Var}}
\def\E{\mathbb{E}}

\def\cov{\text{Cov}}
\def\iid{\stackrel{iid}{\sim}}
\def\df{\mathcal{D}}

\newtheorem{theorem}{Theorem}
\newtheorem{corollary}{Corollary}
\newtheorem{condition}{Condition}
\newtheorem{prop}{Proposition}
\newtheorem{lemma}{Lemma}
\theoremstyle{definition}

\theoremstyle{remark}
\newtheorem*{remark}{Remark}

\numberwithin{equation}{section}
\newcommand{\indep}{\raisebox{0.05em}{\rotatebox[origin=c]{90}{$\models$}}}

\begin{document}

\def\spacingset#1{\renewcommand{\baselinestretch}%
{#1}\small\normalsize} \spacingset{1}
\RestyleAlgo{ruled}


\if1\blind
{
  \title{\bf Scalable and Efficient Hypothesis Testing with Random Forests}
  \author{Tim Coleman, Wei Peng, \& Lucas Mentch \vspace{5mm} \\  Department of Statistics \\ University of Pittsburgh}
  \maketitle
} \fi

\if0\blind
{
  \bigskip
  \bigskip
  \bigskip
  \begin{center}
    {\LARGE\bf Scalable and Efficient Hypothesis Testing with Random Forests}
\end{center}
  \medskip
} \fi

\bigskip
\begin{abstract}
\noindent Throughout the last decade, random forests have established themselves as among the most accurate and popular supervised learning methods.  While their black-box nature has made their mathematical analysis difficult, recent work has established important statistical properties like consistency and asymptotic normality by considering subsampling in lieu of bootstrapping.  Though such results open the door to traditional inference procedures, all formal methods suggested thus far place severe restrictions on the testing framework and their computational overhead often precludes their practical scientific use.  Here we propose a hypothesis test to formally assess feature significance, which uses permutation tests to circumvent computationally infeasible estimates of nuisance parameters.  We establish asymptotic validity of the test via exchangeability arguments and show that the test maintains high power with orders of magnitude fewer computations. Importantly, the procedure scales easily to big data settings where large training and testing sets may be employed without the need to construct additional models.  Simulations and applications to ecological data where random forests have recently shown promise are provided.  
\end{abstract}

\noindent%
\vfill

\spacingset{1} 
\section{Introduction}
\label{sec:intro}
Advances in computing power and big data collection have produced numerous situations in which complex supervised learning methods can drastically outperform more rigid classical statistical models in terms of predictive accuracy.  Despite these advances, many such models and algorithms are largely impenetrable to traditional statistical analysis.  The random forests algorithm \citep{Breiman2001} is among the relatively few supervised procedures for which formal statistical properties have recently been developed, paving the way for inference procedures. As detailed below, however, methods proposed to this point for assessing variable importance have either been \emph{ad hoc} and susceptible to producing misleading results even in simple settings or have come with severe restrictions on the testing framework while incurring extreme computational overhead.  The primary goal of this paper is to formally develop a statistically valid permutation test approach that maintains high power with orders of magnitude fewer required computations that scales naturally and efficiently to large data settings.

Permutation tests have their roots in the work of \citet{Fisher1937} using contingency tables. The canonical permutation test framework relies on an assumption of exchangeability of observations, at least asymptotically. Given iid samples $\bm{X} = X_1,..., X_n$ and $\bm{Y} = Y_1,..., Y_m$, consider the joined sample, $\bm{Z} = \bm{X} \uplus \bm{Y}$, where $\uplus$ indicates concatenation of datasets and let $\mathcal{G}$ be the group of all permutations of the indices $1,...,N$, for $N = m+n$. Let $T = r(Z_1,... Z_{N})$ be a statistic of interest and let $T_0$ denote the statistic calculated on the original data. A p-value for the hypothesis the null hypothesis  $H_0: X \stackrel{d}{=} Y$ is given by
\[
p = \frac{1}{|\mathcal{G}|}\sum_{g \in \mathcal{G}} I(|T_0| > |T(g\bm{Z})|) := 1-\hat{J}_N(T_0; \bm{Z}) + \hat{J}_N(-T_0; \bm{Z})
\] 
where $\hat{J}_N(\cdot; \bm{Z})$ is the permutation distribution function, often referred to as the conditional distribution. To achieve a test with type I error rate $\alpha$, we reject $H_0$ if $p < \alpha$. \citet{Pesarin2010} note that this p-value is conditionally unbiased, i.e. $P(p < \alpha |\bm{Z}, H_0) \leq \alpha$ and $P( p < \alpha |\bm{Z}, H_1) \geq \alpha$. However, this procedure is \textit{not} typically unbiased for more general hypotheses, such as $H_0^f: \E(f(X_1)) = \E(f(Y_1))$ for some integrable function $f(\cdot)$. As such, many permutation procedures are heuristics for hypotheses like $H_0^f$ that may provide some practical use and intuition, but without verified statistical validity. 

\subsection{Permutation Tests}

Classical work on permutation tests from \citet{Hoeffding1952} and \citet{Lehmann1949} demonstrates the convergence of the permutation distribution to the sampling distribution for a wide variety of test statistics. Much of the modern work has focused on extending permutation tests to situations where the data may not be iid or even exchangeable (e.g. \cite{Romano1990}). Studentization is typically proposed as a means of forcing the sampling distribution of a statistic to converge to a normal distribution to which it is then shown that the permutation distribution also converges. This idea has underpinned results in \citet{Neuhaus1993} and \citet{Janssen2005}, who provide various sufficient conditions for the convergence to the unconditional distribution. 

Permutation tests are exact tests for hypotheses of equal distribution under the assumption of iid sequences, but as noted above, are not necessarily valid for more general hypotheses. Convergence to the unconditional distribution ensures that the permutation distribution can be used for a finite sample exact test of equality of distribution and an asymptotically valid test for more general hypotheses. In this work, we prove results regarding the asymptotic validity of our procedure for more general hypotheses. The individual models (base learners) in supervised ensembles, such as decision trees in a random forest, naturally lend themselves to the permutation framework by being exchangeable in many practical cases. 

\subsection{Related Work on Random Forests}
Decision trees  recursively partition the covariate space and generate predictions by fitting some simple model -- often an average or majority vote -- within each resulting region. Of particular interest are the classical \textbf{C}lassification \textbf{A}nd \textbf{R}egression \textbf{T}rees \citep{Breiman1984}. CART procedures often have low bias, but can overfit the data without careful pruning. Bagging stabilizes the variance by training many individual learners on bootstrap samples.  Random forests \citep{Breiman2001} augment the bagging procedure by introducing auxiliary randomness in the construction of each individual learner, leading to trees with a lower degree of dependence but higher individual variances. Since their introduction, random forests have sustained a long track-record of empirical success in terms of predictive accuracy; see \citet{Fernandez2014} for a recent large-scale comparison in which random forests outperform nearly all competitors.   

Extensions of random forests beyond classification/regression forests have also been proposed. \citet{Meinshausen2006} used the weights learned in a regression forest to perform quantile regression.  \citet{Ishwaran2008} suggested survival forests whose consistency properties were recently studied in \citet{Cui2017}. \citet{Zhu2015} proposed constructing trees within a reinforcement learning framework. Finally, \citet{Athey2016} proposed the unifying framework of generalized random forests, which use random forests weights for general local parameter estimation.

The flexibility offered by random forests makes rigorous analysis of their statistical properties challenging. \citet{Breiman2001} proposed out of bag (oob) measures for variable importance, though a substantial amount of work since has shown that these importance measures are biased towards inflating the importance of correlated variables \citep{Tolocsi2011} or variables with many levels \citep{Strobl2007}. 
Alternative measures have been proposed, for example, in \citet{Altmann2010, Janitza2016} but do not come with statistical validity. These procedures also generally involve training many random forests and are thus often computationally intractable.  \citet{Ishwaran2019} provides confidence intervals for the standard oob measures that are valid whenever the random forest is assumed to be $L_2$ consistent. 

\citet{Wager2014a} applied the infinitesimal jackknife variance estimate developed in \citet{Efron2014} to produce closed form variance estimates for random forest predictions.  \citet{Scornet2015} provided the first consistency results for Breiman's original random forest procedure for additive regression functions.  \citet{Mentch2016} derived the closed form asymptotic distribution for random forest predictions under restrictions on subsample size. More recently, \citet{Wager2018} proved both consistency and asymptotic normality for subsampled random forests for potentially larger subsamples whenever trees are restricted to being built according to honesty and regularity conditions and large numbers of trees are constructed.

The asymptotic normality established in \citet{Mentch2016} was obtained by casting random forests as incomplete, infinite-order U-statistics. In addition to establishing normality and providing the closed form asymptotic variance, the authors also lay out a formal hypothesis testing procedure for evaluating variable importance.  This test, though valid, is quite computationally prohibitive.  The hypotheses are presumed to be evaluated at predefined test locations in some test set $\mathcal{T}$ and whenever $|\mathcal{T}|=N_t >1$, calculating the test statistic involves estimating an $N_t \times N_t$ covariance matrix.  Accurate estimation of the covariance necessitates constructing a very large number of trees and becomes computationally infeasible for more than 20-30 test points.  \citet{Mentch2017} extends the procedure to tests for additivity and provide an alternative approximate test involving random projections that allows the procedure to scale up slightly but with additional computational overhead.  Even employing the potentially more efficient infinitesimal jackknife variance estimate utilized in \citet{Wager2014a}  and \citet{Wager2018} requires the number of trees constructed be at least on the order of $n$ to be valid.  

In contrast, the method we propose here is almost entirely computationally immune to the number of test points. The permutation scheme we employ avoids the need for an explicit covariance estimation and thus does not require a larger number of trees for larger datasets.  Instead, our hypothesis tests provide valid p-values for variable importance while maintaining the same order of computational complexity as the original random forest procedure.  Put simply, if the size and structure of the available data allows for a random forest model to be constructed, our testing procedure can be readily employed.  


The remainder of this paper is laid out as follows. In \autoref{sec:meth}, we give an overview of the testing procedure, and further highlight its benefits over existing methods.  In \autoref{sec:theory}, we present results regarding the statistical properties of the proposed test, namely that it attains validity for the desired hypotheses. In \autoref{sec:sims}, we present simulation studies of the testing procedure for a variety of underlying regression functions. In \autoref{sec:models}, we apply our procedure to multiple ecological datasets where random forests have been successfully employed in applied work. In addition to the main text,  all technical proofs are provided in \autoref{appdx:proofs}, and additional simulations demonstrating the robustness of the proposed procedure are presented in \autoref{appdx:more_sims}.

\section{Overview of the Testing Procedure}
\label{sec:meth}

Consider a sample $\df_n = \{Z_1, Z_2, ..., Z_n\}$, with $Z_i = (\bm{X}_i, Y_i)$ consisting of observations on covariates $\bm{X} = (X_1, ..., X_p) \in \mathcal{X}$ and a response $Y \in \mathcal{Y}$. In this work, it is assumed that $Z_k \stackrel{iid}{\sim} F$ where $F$ is some distribution with support on $\mathcal{X} \times \mathcal{Y}$.  In the regression context, we assume that $Y = m(\bm{x}) + \epsilon$ where $m(\bm{x}) = \E(Y|\bm{X} = \bm{x})$ and $\epsilon$ is an independent noise process, typically with $\E(\epsilon) = 0$ and $\Var(\epsilon) < \infty$. 
The goal of the random forest procedure is to accurately estimate $m(\bm{x})$. Each tree in a random forest is constructed by drawing subsamples of size $k_n < n$, from $\df_n$, drawing a randomization parameter $\xi$ from some distribution $\Xi$, and constructing the randomized decision tree. This process is repeated $B$ times and the random forest prediction at some $\bm{x} \in \mathcal{X}$ is given by
\begin{equation} \label{eqn:randfor}
RF_{B,k_n}(\bm{x}) = \frac{1}{B}\sum_{j=1}^B T_{j,k_n}(\bm{x}; \xi_j; \df_n).
\end{equation}
We can similarly evaluate the RF prediction accuracy at a single fixed test location $\bm{x}$ and with true response value $y$ via its mean squared error 
\[
 MSE_{RF}(\bm{x} ; y, \df_n) =  \bigg(\frac{1}{B} \sum_{j=1}^{B} T_{k_n,j}(\bm{x}) - y\bigg)^2 .
\]
For sufficiently large $B$, \autoref{eqn:randfor} can be made arbitrarily close to $\E_{\Xi}(RF_{B,k_n}(\bm{x})|\df_n) $, where the expectation is taken over the subsampling and randomization distributions, conditional on the data. \autoref{eqn:randfor} can be interpreted as a Monte Carlo approximation to the infinite forest. The effect of this approximation is studied in \citet{Scornet2016}. In what follows, we use the shorthand $X_n \stackrel{d}{\to} \mathcal{N}(\mu, \sigma^2)$ to mean that $X_n \stackrel{d}{\to} Z$ where $ Z \sim \mathcal{N}(\mu, \sigma^2)$. 

Similarly, we can calculate the MSE of a forest at a collection of test points $\mathcal{T} = [(\bm{x}_1, y_1),..., (\bm{x}_{N_t}, y_{N_t})]$ as $MSE_{RF}(\mathcal{T}) = \frac{1}{N_t} \sum_{\ell = 1}^{N_t} MSE_{RF}(\bm{x}_\ell; Y_\ell, \mathcal{D}_n)$. Let $RF^\pi$ be defined similarly to \autoref{eqn:randfor}, but with $\df_n$ replaced by $\df_n^\pi$, where $\df_n^\pi$ replaces some subset of features with an alternate copy drawn independent of $Y$ given the rest of the covariates. To make this concrete, suppose that this subset consists of just a single feature $X_j$.  We can then evaluate whether $X_j$ is important by conducting a test of the following hypotheses
\begin{equation} \label{eqn:hypoMSE}
\begin{split}
H^j_0: \ & \E(MSE_{RF}(\mathcal{T})) = \E(MSE_{RF^\pi}(\mathcal{T})) \\
H^j_1: \ & \E(MSE_{RF}(\mathcal{T})) < \E(MSE_{RF^\pi}(\mathcal{T}))
\end{split}
\end{equation}
where the expectation is taken over the training data and auxillary randomness, but is conditional on $\mathcal{T}$. We call $X_j$ important if we are able to reject $H^j_0$. This definition of importance is model based and therefore less general than alternative definitions such as that utilized in the recent \emph{knockoff} literature \citep{Barber2015,Candes2016}, where a variable $X_j$ is deemed unimportant if
\[
(Y\ \indep \ X_j) \ | \ \bm{X}_{-j} .
\]
It should be noted that the above is neither necessary nor sufficient for $H_0^j$. However, in practice, the test statistic utilized in the knockoff procedure is generally taken as the difference in importance measures between original and knockoff variables and thus the outcome of the procedure itself remains highly model dependent. We also note that our procedure, while it could use knockoffs, does not require knowledge of the distribution of the covariates.

\subsection{Testing Procedure}
In general, if two randomized ensemble methods produce predictions that are similarly accurate, then the permutation distribution of discrepancies in accuracy should be centered around 0. In our particular setting for testing feature significance, we compare the accuracy of two random forests built on different data.  For a given (original) dataset $\df_n$, we first construct $\df_n^\pi$ in such a way so as to remove any dependence of response on these features. However, rather than permuting the data and retraining entire random forests, we first train trees on both $\df_n$ and $\df_n^\pi$ separately, record predictions at the test locations, and then permute the predictions between the forests. The new forests formed at each iteration thus consist of some trees built on the original data and some built with the permuted counterpart so that on average.  In this light, the testing procedure can be seen as directly analogous to a classic permutation test to evaluate equality in distribution across two groups.  This procedure requires only $2B$ trees, regardless of the size of the test set.


Pseudo-code for the permutation test is provided in Algorithm 1. We use $\oplus$ to denote concatenation of data matrices by column, $\uplus$ to denote concatenation by row, and $\ominus$ to denote the removal of columns from a dataset.  In order to prevent p-values exactly equal to 0, we add 1 to the numerator and denominator, ensuring that under $H_0$ the p-values are stochastically larger than uniform random variables. This suffices to make the testing procedure slightly more conservative, but more amenable to potential p-value transforming procedures; see, for example, \citet{Phipson2010} for a more thorough discussion.  Crucially, note that this procedure requires no explicit variance estimation of the $N_t$ predictions made by individual forests, thereby providing a dramatic computational speed-up over existing parametric approaches \citep{Mentch2016,Mentch2017} that require the estimation of a $N_t \times N_t$ covariance matrix. 

\begin{algorithm}[h]
\begingroup\footnotesize
 \KwData{Training data $\df_n$ test sample ($\mathcal{T} = [(\bm{x}_1, y_1),..., (\bm{x}_{N_t}, y_{N_t})]$), specified feature(s) of interest, $\bm{X}_S$, $N_0$ number of permutations to evaluate}
 \KwResult{p-value, $\tilde{p}$ for importance  of $\bm{X}_S$ at points in $\mathcal{T}_n$}
 \textsc{set} number of permutations $n_{perm}$, subsample size $k_n$, and $n_{tree} = B$ \;
 \textsc{define} $\bm{X}_S^\pi$ by permuting the rows of $\df_n$ and selecting the columns corresponding to $\bm{X}_S$ \;
 \textsc{define} $\df_n^\pi = \df_n \ominus \bm{X}_S \oplus \bm{X}_S^\pi $\;
 \For{$i$ in $\{1, ..., B\}$}{
 \textsc{sample} $k_n$ rows from  $\df_{n}$: $\df^{*}_i = \{Z_{i,1}^{*}, ..., Z_{i, k_n}^{*}\}$\;
 \textsc{sample} $k_n$ rows from  $\df^\pi_{n}$: $\df^{*\pi}_i = \{Z_{i,1}^{*\pi}, ..., Z_{i, k_n}^{*\pi}\}$\;
 \textsc{train} trees $T_i(\cdot)$ on $\df^{*}_{i,k_n}$ and $T^\pi_i(\cdot)$ on $\df^{*\pi}_{i, k_n}$\;
 \textsc{predict} at $\mathcal{T}_n$ using $T_i, T^\pi_i$, generating $\bm{T}_i = [T_i(\bm{x}_1), ..., T_i(\bm{x}_{N_t})]$ and $\bm{T}^\pi_i = [T^\pi_i(\bm{x}_1), ..., T^\pi_i(\bm{x}_{N_t})]$
 }
 \textsc{calculate} $MSE_0 = \frac{1}{N_t}\big|\big| \frac{1}{B} \sum_{i=1}^B \bm{T}_i - \bm{y}\big|\big|^2_2$  and $MSE^\pi_0 =  \frac{1}{N_t}\big|\big| \frac{1}{B} \sum_{i=1}^B \bm{T}^\pi_i - \bm{y}\big|\big|^2_2$\;
 \For{$j$ in $\{1, ..., N_0\}$}{
 \textsc{sample} $\bm{T}^*_{j,1}, ..., \bm{T}^*_{j,B}$ from $\{\bm{T}_{1},... \bm{T}_B, \bm{T}^\pi_1, ..., \bm{T}_B^\pi\}$ without replacement, call the $B$ remaining trees $\bm{T}^{*\pi}_{j,1}, ..., \bm{T}^{*\pi}_{j,B}$\;
 \textsc{calculate} $MSE^*_j =  \frac{1}{N_t}\big|\big| \frac{1}{B} \sum_{l=1}^B \bm{T}^*_{j,l} - \bm{y}\big|\big|_2^2$ and  $MSE^{*\pi}_j =  \frac{1}{N_t}\big|\big| \frac{1}{B} \sum_{l=1}^B \bm{T}^{*\pi}_{j,l} - \bm{y}\big|\big|_2^2$
 }
 \textsc{calculate} $\tilde{p} = \frac{1}{N_0 +1}\bigg[1 + \sum_{j=1}^{N_0} I\big((MSE^\pi_0 - MSE_0) \leq (MSE^{*\pi}_j - MSE^{*}_j)\big)\bigg]$
 \caption{Permutation test pseudocode for variable importance} \label{alg:alg1}
 \endgroup
\end{algorithm}


\section{Theoretical Justification} \label{sec:theory}

We now develop the theoretical backing for the hypothesis testing procedure outlined above. In \autoref{subsec:exchange}, we make explicit the connection between bagged-models and exchangable random variables. Then, in \autoref{subsec:trees}, we use these results establish asymptotic normality for subsampled random forest predictions, under mild conditions. Next, in \autoref{subsec:mse}, we extend these results to the fixed-test set MSE, establishing a CLT for this quantity. To use these distributions directly is unwieldy - there is no obvious consistent estimator of the variance parameters available. Thus, in \autoref{subsec:permtests}, we prove that the proposed permutation test is asymptotically equivalent to the computationally infeasible parametric test, building on recent arguments from \citet{Chung2013}.  For readability, most proofs are reserved for \autoref{appdx:proofs}.

\subsection{Exchangeable Random Variables \& Permutation Tests} \label{subsec:exchange}

Recall that a sequence of random variables $X_1, X_2, ...$ is exchangeable if $(X_{i_1}, X_{i_2}, ...., X_{i_k}) \stackrel{d}{=} (X_{\pi(1)}, X_{\pi(2)}, ..., X_{\pi(k)})$ for every finite sub-collection indexed by $i_1, ..., i_k$ and every permutation of the indices $\pi(\cdot)$, see \citet{Aldous1985} for a thorough review. 


Permutation tests naturally lend themselves to exchangeable data by providing a means of evaluating the hypothesis that the joint distribution of a collection of random variables is invariant under permutations. They maintain exactness for the null hypothesis whenever $X_i \iid P$ and independently $Y_j \iid Q$ because the joint measure of the data factorizes as
\[
\mu(\bm{X}, \bm{Y}) =  \prod_{i=1}^n P(X_i) \prod_{j=1}^m Q(Y_j) 
\]
which is invariant to permutations of observations if and only if  $P = Q$.

Modern work for permutation tests has focused largely on modifications needed to account for violations of the exchangeability assumption. \citet{Chung2013} propose a studentization of the permutation test statistic when conducting inference a functional of two distributions. Consider, for example a two sample problem, with $X_1, ..., X_n \stackrel{iid}{\sim}P_X = \mathcal{N}(0, 5)$  and independently let $Y_1, ..., Y_m \stackrel{iid}{\sim}P_Y = \mathcal{N}(0,1)$. Clearly, $\text{median}(P_X) = \text{median}(P_Y)$, but the data are no longer exchangeable and so an unstudentized permutation test of $H_0:\text{median}(P_X) = \text{median}(P_Y)$ is no longer valid at a pre-specified level (see \citet{Chung2013} for details). However, note that exchangeability is violated only because the data are no longer identically distributed; permutation tests can remain valid for data that are correlated but identically distributed so long as the pairwise dependence is constant. \\

\vspace{-4mm}
\begin{theorem} \label{theorem:exch}
Denote a sequence of (potentially randomized) subsampled trees as $\{T_k(\cdot)\}_1^\infty$. Under the conditions outlined above, the residuals at $\bm{Z}^* = (\bm{X}^*, Y^*) \sim F$ given by 
\[
r_k = T_k(\bm{X}^*) - Y^*
\]
form an infinitely exchangeable sequence of random variables.
\end{theorem}


In the case of a single random forest, exchangeability is readily apparent as the order in which trees are trained has no bearing on their structure. Indeed, \autoref{theorem:exch} can be extended to any bagged learning method. 

The primary goal of this work is to identify covariates that produce statistically significant improvements in model accuracy. To assess this, we consider building two forests, one on the original dataset $\df_n$ and another on a second dataset $\df^\pi_n$ wherein the covariate(s) of interest $\bm{X}_S$ are rendered independent of $Y$, conditional on the rest of the features.  This muting can be achieved in various ways:
\begin{itemize}
\item Outright exclusion:  $\bm{X}_S$ is simply removed from the second training dataset.
\item Random permutation:  Each covariate in $\bm{X}_S$ is randomly shuffled so that $\bm{X}_S$ is replaced by some permuted alternative $\bm{X}_{S}^{\pi}$ in the second training dataset.
\item Knockoffs:  Each covariate in $X_i$ in $\bm{X}_S$ is replaced by some knockoff alternative $X_{i}^{\pi}$ sampled from the distribution of $X_i | \bm{X}_{-i}$ so that $\bm{X}_S$ is replaced by a randomized alternative $\bm{X}_{S}^{\pi}$ in the second training dataset. See \citet{Candes2016} for details.
\end{itemize}


Given two samples $\df_n^1$ and $\df_n^2$ drawn independently from the same population and a collection of subsampled trees, say $\bm{T}_1$ and $\bm{T}_2$, trained on each, $\bm{T}_1 \stackrel{d}{=} \bm{T}_2$, but the trees are no longer exchangeable because the conditioning random vector (i.e.\ the training data) is different.  In general, the within sample dependence between trees will be higher than the between sample correlation. In practice, the trees are \textit{approximately} exchangeable - higher within sample dependence tends to die out when the subsample size grows slower than $n$. More specifically, provided the tree distributions converge weakly to some distribution, we can establish the notion of \textit{asymptotic exchangeability}. This idea is well studied, if rarely explicitly mentioned, with necessary conditions provided, for example, in \citet{Romano1990, Chung2013,Good2002}. Recall also that the original motivation for inserting the additional randomness in random forests was to reduce between-tree correlation, further dampening the effect. 

In practical terms, this implies that replacing the covariates under investigation with a randomized counterpart -- either a permutation or knock-off -- leads to a procedure with more desirable properties than when those covariates are simply dropped, creating a lower-dimensional covariate space. In particular, to attain near-exchangeability under the null hypothesis, the individual trees should be constructed in a near-identical fashion by building trees on data with the same dimension feature space, on the same subsample size, and with other various tree-specific parameters controlled.

Given a dataset $\df_n$ with $n\times p $ design matrix $\bm{X}$, let $\mathcal{S} \subset \{1,...,p\}$ and define $\bm{X}_\mathcal{S} = \{ X_j : \ j \in \mathcal{S}\}$ and $\bm{X}_{-\mathcal{S}} = \{X_j : \ j \notin \mathcal{S}\}$ where we take $\bm{X}_\mathcal{S}$ to be the covariates of interest. We then create a randomized version of $\bm{X}_\mathcal{S}$ independent of $Y$, denoted by $\bm{X}_\mathcal{S}^{\pi}$. Note in particular that when the entire joint density $P(\bm{X})$ of the covariates is known, Algorithm 1 of \citet{Candes2016} can be used to generate the knockoffs that make up $\bm{X}_\mathcal{S}^{\pi}$ which then ensures that $[\bm{X}_{-\mathcal{S}}, \bm{X}_\mathcal{S}] \stackrel{d}{=}[\bm{X}_{-\mathcal{S}}, \bm{X}_\mathcal{S}^{\pi}] $. By construction, $\bm{X}^\pi_\mathcal{S} \ \indep \ Y | \bm{X}_{-\mathcal{S}}$ and consequently, if we now replace $\bm{X}_\mathcal{S}$ with $\bm{X}_\mathcal{S}^{\pi}$ in the design matrix to form a new training dataset $\df^\pi_n$, then the trees trained on $\df^\pi_n$ inherit the conditional independence so that $T(\bm{x} ; \df^\pi_n) \  \indep \ Y | \ \bm{X}_\mathcal{-S}$. 

Assuming $\bm{X}_\mathcal{-S}$ are the only predictively important covariates, we would expect predictions from trees trained on $\df_n$ to have the same distribution as those generated from trees trained on $\df^\pi_n$. As such, the trees should be asymptotically exchangeable between forests and it follows that we can test this exchangeability assumption via a permutation test.




\subsection{Asymptotic Behavior of Trees} \label{subsec:trees}
Within-forest exchangeability is not sufficient to justify the proposed testing procedure at the nominal level. Instead, we need to establish sufficient conditions to justify exchanging trees between forests. An important step in this direction is to establish the existence of a limiting sequence of subsampled trees that behave like an iid sequence. \\

\vspace{-4mm}
\begin{condition} \label{cond:cond1}
There exists a random function $T_\infty$ such that $\lim_{n\to\infty} T_{k_n} \stackrel{d}{=} T_\infty$ 
\end{condition}

In \autoref{subsec:treespec} we provide sufficient conditions for this to be satisfied. Note that this condition is similar in spirit to Assumption 15.7.1 in \citet{Lehmann2006}, which is fundamental to the validity of subsampling based intervals for model parameters. 


In practice, we would like to establish results for random forests trained on growing subsamples. If we insist that the subsample size $k_n$ grow slower than $\sqrt{n}$, we obtain the following intuitive result. \\

\vspace{-3mm}
\begin{lemma} \label{lem:lemma1}
Consider a collection of $B_n$ trees  built from a training dataset of size $n$ on subsamples of size $k_n$, say $\{T_{j, k_n}\}_{j=1}^{B_n}$, satisfying \autoref{cond:cond1}. Then, as long as $k_n/\sqrt{n} \to 0$ and
\[
 \binom{B_n}{2} \log\bigg[\frac{\binom{n- k_n}{k_n}}{\binom{n}{k_n}}\bigg] \to 0
\]
the infinite sample sequence of trees, $\{T_{1,\infty, k_\infty}, ..., T_{B,\infty, k_\infty},...\}$, is an infinite sequence of pairwise independent random functions.
\end{lemma}
 
The condition on the number of trees $B_n$ is likely not of much practical importance. For finite $B_n$, the probability sequence has the form of $a_n^K$, so because $a_n \to 1$, $a_n^K$ also converges to 1. However, if we let $B_n$ grow with $n$, the number of trees may overwhelm the independence induced by subsampling. Thus, we must let the log probability of an individual pair being independent go to 0 faster than $\binom{B_n}{2} \approx B_n^2/2$ goes to infinity.

\autoref{lem:lemma1} establishes asymptotic pairwise independence, but not that the limiting sequence is iid. For this, we turn to a result from \citet{Aldous1985}. \\

\vspace{-4mm}
\begin{lemma} \label{lem:lemma2}
\citep{Aldous1985} Let $Z_1, Z_2, ...$ be an infinitely exchangeable sequence. If $Z_i\ \indep\ Z_j, i \neq j$, then $Z_1, Z_2, ...$ is a sequence of iid random variables.
\end{lemma}

An immediate consequence of the preceding lemmas is the following corollary. \\

\vspace{-4mm}
\begin{corollary} \label{cor:corl1}
Let $\{T_{j, k_n}\}_{j=1}^{B_n}$ be a collection of $B_n$ trees trained on subsamples from $\df_n$, satisfying the conditions of \autoref{lem:lemma1}. Then, $\{T_{j, \infty}\}_{j=1}^\infty := \lim_{n\to\infty}\{T_{j, k_n}\}_{j=1}^{B_n} $ is an iid sequence of functions. 
\end{corollary}

The infinite sequence of subsampled trees enjoys many properties that the finite sequence does not. In particular, we can obtain the following pointwise central limit theorem. \\ 

\vspace{-4mm}
\begin{corollary} \label{cor:CLT}
Let $\{T_{j, k_n}\}_{j=1}^{B_n}$ be a sequence of trees on subsamples from $\df_n$, satisfying the conditions of \autoref{lem:lemma1} and \autoref{cond:cond1}. Further, assume $\bm{x} \in \mathcal{X}$ is such that $0 < \Var(T_\infty(\bm{x})) = \sigma^2(\bm{x}) < \infty $. Then as $n\to\infty$
\begin{equation} \label{eqn:CLT}
\sqrt{B_n}\bigg[\frac{1}{B_n}\sum_{i=1}^{B_n}T_{i,k_n}(\bm{x}) - \mathbb{E}\bigg(\frac{1}{B_n}\sum_{i=1}^{B_n}T_{i,k_n}(\bm{x})\bigg)\bigg] \stackrel{d}{\to} \mathcal{N}(0, \sigma^2(\bm{x}))
\end{equation}
\end{corollary}

\autoref{cor:CLT} follows directly from applying the Central Limit Theorem to the sequence of univariate random variables $\{T_{j, \infty}(\bm{x})\}_{j=1}^\infty$, which are iid by \autoref{cor:corl1}.

\begin{remark}  For a collection of test points, $\bm{x}_1,..., \bm{x}_{N_t}$, we can also consider the sequence of vectors $\bm{T}_{i,k_n} = [T_{i,k_n}(\bm{x}_1), ..., T_{i,k_n}(\bm{x}_{N_t})]^T$, which are iid by \autoref{cor:corl1}. If we assume that $\Sigma = \E\big[(\bm{T}_{i,k_n} - \E(\bm{T}_{i,k_n}))(\bm{T}_{i,k_n} - \E(\bm{T}_{i,k_n}))^T\big] $ has finite entries, the multivariate central limit theorem gives that as $n\to\infty$
\[
\sqrt{B_n}\bigg[\frac{1}{B_n}\sum_{i=1}^{B_n}\bm{T}_{i,k_n} - \mathbb{E}\bigg(\frac{1}{B_n}\sum_{i=1}^{B_n}\bm{T}_{i,k_n}\bigg)\bigg] \stackrel{d}{\to} \mathcal{N}(0, \Sigma).\]
\end{remark}

\vspace{3mm}

\begin{remark} We can generalize the independence results to a collection of two sets of trees. In particular, suppose that we now train $B_n/2$ trees on $\df_n = \{Z_i\}_{i=1}^n$ and $\df_n^\pi = \{Z_i^\pi\}_{i=1}^n$, where $Z_i^\pi = ( [\bm{X}_\mathcal{S}, \bm{X}^\pi_{-\mathcal{S}}]_i ,Y_i)$. Note that $Z_i^\pi\ \indep\ Z_j, \forall \ i \neq j$, so there is the same independence structure between the datasets as within. Thus, the probability that a pair of trees trained on subsamples of size $k_n$, one from  $\df_n$ and one from $\df_n^\pi$, are independent is the same as the probability that a pair of trees within forest are independent. As such, $\{T_{i,k_n}(\bm{x})\}_{i=1}^{B_n}$ and $\{T^\pi_{i,k_n}(\bm{x})\}_{i=1}^{B_n}$, where $B_n, k_n$ satisfy the conditions of \autoref{lem:lemma1}, behave like two independently iid samples.
\end{remark}

We intentionally leave $\sigma(\bm{x})$ as an abstraction since estimation of $\sigma(\bm{x})$ is not straightforward. Instead, this result will be used as the basis for asymptotic validity of our permutation test which, uncharacteristically, is far more computationally efficient. Going forward, we consider the asymptotic case, so that the sequence of tree predictions behaves like an iid sequence. Further, in the infinite sample case, the number of trees can be made arbitrarily large, and so we allow $B$ to go to infinity with the understanding that it does so in such a way that respects the requirements of \autoref{lem:lemma1}.  This is largely a matter of notational convenience; we could explicitly include the dependence on $n$ in each of the following statements and stress that the limiting distributions only hold as $n\to\infty$.

\subsection{Asymptotic Distribution of MSEs} \label{subsec:mse}
Unfortunately, the MSE is not a linear function of exchangeable random variables and thus requires more careful attention before being used as a test statistic in a permutation test.  In this subsection we establish the asymptotic normality of the MSE which we then utilize to show that the difference in MSEs between two forests is asymptotically normal.  To begin, consider a single test point $(\bm{x},y)$. We can write the MSE as
\begin{equation} \label{eqn:MSERF}
 MSE_{RF}(\bm{x} ; y) = g\left( \frac{1}{B}\sum_{i=1}^B T_{i}(\bm{x}), y  \right)
\end{equation}
where $g(a, b) = (a-b)^2$. In what follows, we suppress the dependence on $y$, writing just $ MSE_{RF}(\bm{x} ; y) = g\left(RF_B(\bm{x})\right)$. We derive the asymptotic distribution of the MSE via the delta method, which we belabor here for its intuitive value. We can then appeal to the mean value theorem to say
\[
g(RF_B(\bm{x})) = g(\E RF_B(\bm{x})) + g'(\tilde{R}_B(\bm{x}))[RF_B(\bm{x}) - \E RF_B(\bm{x})]
\]
where $\tilde{R}_B(\bm{x})$ is a random quantity bounded between $RF_B(\bm{x}), \E RF_B(\bm{x})$. The law of large numbers gives that $RF_B(\bm{x}) = \E RF_B(\bm{x}) + o_P(1) $ and further $\tilde{R}_B(\bm{x}) \stackrel{p}{\to} \E RF_B(\bm{x})$. Next, continuity of $g'$ gives that  $g'(\tilde{R}_B(\bm{x})) \stackrel{p}{\to} g(\E RF_B)$. Thus,
\[
\begin{split}
\sqrt{B}\left[g(RF_B(\bm{x})) - g(\E RF_B(\bm{x})) \right] &= g'(\tilde{R}_B(\bm{x}))\sqrt{B}\left[ RF_B(\bm{x}) - \E RF_B(\bm{x})\right] \\ &\stackrel{d}{\to} \mathcal{N}\left(0, g'(\E RF_B(\bm{x}))^2 \sigma^2 \right) \\
& \stackrel{d}{=} \mathcal{N}\left(0, 4(\E RF_B(\bm{x}) - y)^2 \sigma^2 \right) \ \text{for $g(x) = (x-y)^2$}.
\end{split}
\]
The calculation above is more informative - we see that the MSE is asymptotically a linear function of the random forest prediction. An issue is that the above quantity is centered around $g(\E RF_B(\bm{x}))$ rather than $\E g(RF_B(\bm{x}))$, which we now address. In particular, suppose we begin by centering around $\E g (RF_B(\bm{x}))$ rather than $g(\E RF_B(\bm{x}))$. Then, 
\begin{multline} \label{eqn:DeltaCenter}
    \sqrt{B}\left[ g\left(RF_B(\bm{x}) \right) - \E g\left(RF_B(\bm{x})\right)\right] = \\ \sqrt{B}\left[ g\left(RF_B(\bm{x}) \right) - g(\E RF_B(\bm{x}))\right] + \sqrt{B} \left[ g\left(\E RF_B(\bm{x})\right) - \E g\left(RF_B(\bm{x})\right)\right]
\end{multline}
so that if $ \sqrt{B} \left[ g\left(\E RF_B(\bm{x})\right) - \E g\left(RF_B(\bm{x})\right)\right]= o(1)$, then the same distributional result holds. This is shown in \autoref{lem:centered}.
\begin{lemma} \label{lem:centered}
Assume the conditions needed from \autoref{cor:CLT}. Additionally, assume that $g$ has at least $k$ derivatives for some $k \geq 3$ , and that $g^{(k)}(x) < \infty$ for all $x$. Further, assume that $\E | T_i(\bm{x})|^k <\infty$.   Then, 
\[
\sqrt{B}\left[\E g(RF_B(\bm{x})) - g(\E RF_B(\bm{x})) \right]= \frac{g''(\E RF_B(\bm{x}))\sigma^2}{2\sqrt{B}} + o(B^{-3/2}) = o(1)
\]
\end{lemma}
Since the MSE function defined as $g(RF_B(\bm{x})) = \left( RF_B(\bm{x}) - y \right)^2$ satisfies the conditions posited by \autoref{lem:centered}, we can conclude that
\[
 \sqrt{B}\left[ g\left(RF_B(\bm{x}) \right) - \E g\left(RF_B(\bm{x})\right)\right] \stackrel{d}{\to}   \mathcal{N}\left(0, g'(\E RF_B(\bm{x}))^2 \sigma^2 \right).
 \]

Application of the mean value theorem requires that $g'(\E RF_B(\bm{x})) \neq 0$ if and only if $\E RF_B \neq y$. The expected prediction can be written as $\E RF_B(\bm{x}) = m(\bm{x}) + \delta(\bm{x})$, where $\delta(\bm{x})$ is the pointwise bias of the random forest. Recalling that the response is given by $Y = m(\bm{x}) + \epsilon$, if it holds for all $\bm{x}$ that $P(\epsilon \neq \delta(\bm{x})) = 1$, then the result holds for the squared error calculated with respect to almost all $Y$ and thus is trivially satisfied for continuous errors.  A similar result could be applied to any continuously differentiable loss function $g(\cdot, \cdot)$, again under the condition that $g'$ is almost surely non zero. 

\begin{remark}
This CLT result does not depend on \autoref{cor:CLT}. In fact, a similar argument could be used to justify the asymptotic normality of the MSE for any random forest who satisfies a central limit theorem and a law of large numbers (with respect to its own expectation), such as the results in \citet{Mentch2016} and \citet{Wager2018}.
\end{remark}

We can extend this result to the two forest case, where we compare the MSE of $RF_B(\bm{x})$ against that of $RF_B^\pi(\bm{x})$. In particular, if $\E MSE_{RF}(\bm{x}; y) = \E MSE_{RF^\pi}(\bm{x};y)$, we see that 
\begin{equation}\label{eqn:uncond}
\sqrt{B}\left[ MSE_{RF}(\bm{x}; y)  - MSE_{RF^\pi}(\bm{x}; y)  \right] \stackrel{d}{\to} \mathcal{N}\left(0, g'(\E RF_B(\bm{x}))^2\sigma^2 + g'(\E RF^\pi_B(\bm{x}))^2\sigma^2_\pi\right)
\end{equation}

where $\sigma^2_\pi = \Var(T^\pi(\bm{x}))$. This extension uses a similar argument as before to justify centering around the expected MSE instead of the MSE of the expected forest.

Now consider a test set with many points, denoted $\mathcal{T} = [(\bm{x}_1, y_1), ..., (\bm{x}_{N_t}, y_{N_t})]^T$. Given the vector of random forest predictions, $RF_B(\mathcal{T})$, we can calculate the pointwise squared errors as $ MSE_{RF}(\mathcal{T}) = \left[(RF_B(\bm{x}_i) - y_i)^2\right]_{i=1}^{N_t}$. We now introduce some additional notation to establish joint asymptotic normality of this quantity. Let $g_m: \mathbb{R}^m \times \mathbb{R}^m \mapsto \mathbb{R}^m$ be defined as $g_m( [(a_1, b_1),..., (a_m, b_m)]^T) = [g(a_1, b_1),..., g(a_m, b_m)]^T$ for some continuously differentiable function $g(\cdot, \cdot)$.  In particular, we can write the MSE vector as $MSE_{RF}(\mathcal{T}) = g_{N_t}\left(\left[\left(RF_B(\bm{x}_1), y_1\right), ..., \left(RF_B(\bm{x}_{N_t}), y_{N_t}\right)\right] \right)$. 
Then, recalling that a multivariate CLT was shown to hold, we can again appeal to the mean value theorem to write
\[
\sqrt{B}\left(MSE_{RF}(\mathcal{T}) - \E MSE_{ RF}(\mathcal{T}) \right) = \sqrt{B} \nabla g_{N_t}( \tilde{R}_B(\mathcal{T}))\left( RF_B(\mathcal{T}) -\E RF_B(\mathcal{T}) \right) 
\]
where $\tilde{R}_B(\mathcal{T})$ is some point in the hyper-rectangle defined by $\bigotimes_{i=1}^{N_t}[RF_B(\bm{x}_i), \E RF_B(\bm{x}_i)]$, where the intervals are understood to begin at $\text{min}\{RF_B(\bm{x}_i) , \E RF_B(\bm{x}_i)\}$ and $\bigotimes$ is the cartesian product. Note that the area of this rectangle vanishes, so that $\sqrt{B}\left(MSE_{RF}(\mathcal{T}) - \E MSE_{ RF}(\mathcal{T}) \right)$ is also asymptotically a linear rescaling of the random forest prediction. Therefore, the multipoint MSE satisfies
\[
\sqrt{B}\left(MSE_{RF}(\mathcal{T}) -  \E MSE_{ RF}(\mathcal{T}) \right) \stackrel{d}{\to} \mathcal{N}\left(0, [\nabla g_{N_t}(\E RF_B(\mathcal{T}))]^T \Sigma  [\nabla g_{N_t}(\E RF_B(\mathcal{T}))] \right)
\]
as $B\to\infty$ where $\Sigma$ is the covariance induced by the proximity of the test points in $\mathcal{T}$. Then, let $\bm{1/N_t}$ be the $N_t \times 1$ vector of the quantity $1/N_t$, so that as $B\to\infty$,  the limiting distribution of the MSE of a subsampled random forest is given by
\begin{multline} \label{eqn:MSEdistn}
\sqrt{B}(\bm{1/N_t}^TMSE_{RF}(\mathcal{T}) - \bm{1/N_t}^T\E MSE_{RF}(\mathcal{T}))) \stackrel{d}{\to}\\ \mathcal{N}\left( 0, \bm{1/N_t}^T [\nabla g_{N_t}(\E RF_B(\mathcal{T}))]^T\Sigma [\nabla g_{N_t}(\E RF_B(\mathcal{T}))]  \bm{1/N_t}\right)
\end{multline}
where we have appealed to the Cram{\' e}r-Wold theorem in the above.

Finally, to connect back to the testing procedure proposed earlier, we now derive the asymptotic distribution of the differences in MSE between two forests. Let $MSE_{RF}(\mathcal{T})$ be the MSE of a random forest at a set of test points $\mathcal{T}$ and let  $MSE_{RF^\pi}(\mathcal{T})$ denote the MSE of a forest trained on the partially randomized data. By the results above, under the hypothesis that $\E MSE_{RF}(\mathcal{T}) = \E MSE_{ RF^\pi}(\mathcal{T})$, we have that as $B\to\infty$,
\[
\sqrt{B}\bm{1/N_t}^T(MSE_{RF}(\mathcal{T}) - MSE_{RF^\pi}(\mathcal{T})) \stackrel{d}{\to} \mathcal{N}(0, \tau^2)
\]
for some $\tau^2 > 0$ that does not necessarily have a form that is amenable to analysis. To calculate $\tau^2$, first note that for each MSE $g_j(RF_B(\bm{X}_j)) = (RF_B(\bm{X}_j) - Y_j)^2$, by continuity, $g'_j(\tilde{R}_B(\bm{X}_j)) = g'_j(\E RF_B(\bm{X}_j)) + o_P(1)$. Thus, we see that
\[
\begin{split}
MSE_{RF}(\mathcal{T}) -\E MSE_{RF}(\mathcal{T})&= \frac{1}{N_t} \sum_{j=1}^{N_t} MSE_{RF}(\bm{X}_j, Y_j)\\
 &= \frac{1}{N_t} \sum_{j=1}^{N_t}g'_j(\E RF_B(\bm{X}_j))\left[RF_B(\bm{X}_j) - \E RF_B(\bm{X}_j)\right] + o_P(1) \\
 &= \frac{1}{N_t} \sum_{j=1}^{N_t}g'_j(\E RF_B(\bm{X}_j))\left[\frac{1}{B}\sum_{i=1}^B\left[T_i(\bm{X}_j) - \E RF_B(\bm{X}_j)\right]\right] + o_P(1) \\
&= \frac{1}{B}\sum_{i=1}^B \underbrace{\frac{1}{N_t}\sum_{j=1}^{N_t} g'_j(\E RF_B(\bm{X}_j))\left[T_i(\bm{X}_j) - \E RF_B(\bm{X_j})\right]}_{\bar{T}_i} + o_P(1)
\end{split}
\]
where $g_j(\cdot)$ is used to suggest that the squared difference is calculated with respect to $Y_j$.  $\bar{T}_i$ is an iid sequence, so that $\sqrt{B}[MSE_{RF}(\mathcal{T}) - \E MSE_{RF}(\mathcal{T})]$ is asymptotically an iid sum with mean 0 and variance $\sigma^2_{\bar{T}}$ given by
\begin{equation} \label{eqn:sigmabar}
\sigma^2_{\bar{T}} = \frac{1}{N_t}\sum_{j=1}^{N_t} \sigma^2_j \left( g'_j(\E RF_B(\bm{X}_j))\right)^2 + \frac{2}{N_t} \sum_{i < j}g'_j(\E RF_B(\bm{X}_j))g'_i(\E RF_B(\bm{X}_i)) \rho_{ij}
\end{equation}
where $\rho_{ij} = \cov(T(\bm{X}_i), T(\bm{X}_j))$ and $\sigma^2_j = \Var(T(\bm{X}_j))$. We can obtain a similar variance ($\sigma^2_{\bar{T}^\pi}$) for $MSE_{RF^\pi}(\mathcal{T})$, so that under the hypothesis that $\E MSE_{RF}(\mathcal{T}) = \E MSE_{ RF^\pi}(\mathcal{T})$, $\tau^2$ can be seen to be 
\[
\tau^2 = \sigma^2_{\bar{T}} + \sigma^2_{\bar{T}^\pi}.
\]
That the $\bar{T}_i$ and $\bar{T}^\pi_i$ are iid follows from \autoref{lem:lemma2}. Independence of the two samples follows from a similar argument to the second remark after \autoref{cor:CLT}. Crucially, there are many complicated quantities in this \autoref{eqn:sigmabar}, i.e. $\sigma^2_j, \sigma^2_{\pi, j}, \rho_{ij}, \rho^\pi_{ij}$, for which there are not obvious estimators available and thus this result alone is not clearly practical. In the next following sections, we verify the validity of our proposed permutation procedure, which avoids the necessary explicit estimation of these quantities.


\subsubsection{Tree-specific results}\label{subsec:treespec}
Until now, our discussion has remained largely agnostic to the type of base-learners employed, subject to the regularity conditions needed for asymptotic normality. We now argue that the trees typically grown in a random forest satisfy such conditions.  The following result follows a similar strategy as Lemma 2 in \citet{Meinshausen2006} with regularity conditions similar to those imposed in \citet{Wager2018}. \\

\vspace{-4mm}
\begin{prop} \label{prop:prop1}
Assume that $Y = m(\bm{X}) + \epsilon$, where $m(\cdot)$ is continuous on the unit cube. Let $\mathcal{X} = [0,1]^p$, and assume that $X_{i,j} \iid Unif(0,1)$ for $i = 1,..., n$ and $j = 1,...,p$. Then, let $T_n(\bm{x})$ be a tree trained on iid pairs $(\bm{X}_1, Y_1),..., (\bm{X}_n, Y_n)$ such that each leaf of the tree contains a single observation. Further, assume the trees satisfy the following two conditions:

\begin{enumerate}[(i)]
\item $\exists \ \gamma > 0$ such that $P(\text{variable $j$ is split on}) > \gamma$ for $j \in \{1,...,p\}$
\item Each split leaves at least $\gamma n$ observations in each node.
\end{enumerate}

Then, for each $\bm{x}  \in \mathcal{X}$
\[
T_n(\bm{x}) \stackrel{d}{\to} Y|\bm{X} = \bm{x} \ \text{as} \ n \to \infty.
\]
\end{prop}

The tree predictions thus asymptotically behave like the conditional samples of $Y$ and as a result, should have finite non zero variance. Note that \citet{Breiman2001} recommends building trees to full depth in which case \autoref{cond:cond1} is automatically satisfied.

\subsection{Extension to Permutation Tests} \label{subsec:permtests}

In \autoref{subsec:mse} we established that the sampling distribution of MSE differences between forests was asymptotically Gaussian, but with a computationally intractable variance. Here we show that the permutation distribution converges to that sampling distribution.  We begin by restating a classical theorem from Hoeffding. \\

\vspace{-4mm}
 \begin{theorem} \label{theorem:hoeffding}
\citep{Hoeffding1952}    For a sequence of data $\{X_i\}_{i=1}^N$ and a statistic $S: \mathbb{R}^N \to \mathbb{R}$, define the permutation distribution function as \[\hat{J}_N(t) = \frac{1}{|\mathcal{G}_N|}\sum_{\pi \in \mathcal{G}_N} I\big\{S(X_{\pi(1)}, ... , X_{\pi(N)}) \leq t\big\}
    \]
where $\mathcal{G}_N$ is the group of all permutations of $\{1,...,N\}$. Let $\pi, \pi'$ be two permutations drawn independently and uniformly over $\mathcal{G}_N$, and suppose that as $N\to\infty$
    \begin{equation} \label{eqn:joint_rand}
    \big(S(X_{\pi(1)}, ... , X_{\pi(N)}), S(X_{\pi'(1)}, ... , X_{\pi'(N)})\big) \stackrel{d}{\to} (S, S')
        \end{equation}
where $S, S'$ are iid with cdf $R(\cdot)$. Then for all $t$ at which $R(\cdot)$ is continuous, $\hat{J}_N(t) \stackrel{p}{\to} R(t)$.
    \end{theorem}


Direct application of \autoref{theorem:hoeffding} is often challenging. Suppose $\{ X_i\}_{i=1}^n \iid P_X $ and independently $\{Y_i \}_{i=1}^m  \iid P_Y$, and we calculate the statistic $\sqrt{n+m}\left[S(X_1,..., X_n)  - S(Y_1,..., Y_m)\right]$, and further define $p = \lim_{n\to\infty} \frac{n}{n+m}$. Theorem 2.1 of \citet{Chung2013} states that if there exists a function $\psi_{P_Z}$ (which may depend on the distribution  of the data, $P_Z$) such that 
\begin{equation} \label{eqn:asymp_linear}
\sqrt{N}\left[S(Z_1,..., Z_N) - \E S(Z_1, ..., Z_N)  \right] = \frac{1}{\sqrt{N}}\sum_{i=1}^N \psi_{P_Z}(Z_i) + o_{P_Z}(1)
\end{equation}
(i.e.\ the statistic is asymptotically linear), then the permutation distribution of the aforementioned statistic is asymptotically normal with mean 0 and variance given by 
\begin{multline} \label{eqn:tausqd}
\tau^2 = \frac{1}{p(1-p)} \Var(\psi(Z) \ | \ Z \sim pP_X + (1-p)P_Y) \\
= \frac{1}{p(1-p)} \left[p \Var( \psi(X)) + (1-p) \Var(\psi(Y)) \right].
\end{multline}
A key challenge is that $\tau^2$ is often not equal to the variance of the unconditional distribution without additional assumptions on $P_X$ and $P_Y$. 
 A canonical example of this phenomenon is the permutation distribution of the difference in sample means. Given two independent iid samples $X_1,..., X_n$ and $Y_1,..., Y_m$, with $\Var(X) = \sigma^2_X < \infty$, $\Var(Y) = \sigma^2_Y < \infty$, and $\E X = \E Y$, the central limit theorem gives that $\sqrt{n + m}\left(\bar{X}_n - \bar{Y}_m \right) \stackrel{d}{\to} \mathcal{N}\left(0, \frac{1}{p}\sigma^2_X  + \frac{1}{1-p}\sigma^2_Y\right)$ where $p = \lim_{n\to\infty} \frac{n}{n+m}$. The conclusion of \autoref{eqn:tausqd}, however, is that the permutation distribution of the statistic $\sqrt{n+m}\left(\bar{X}_n - \bar{Y}_m \right) $ approaches a normal disribution with mean 0 and variance $\frac{1}{1-p}\sigma^2_X + \frac{1}{p}\sigma^2_Y$ \citep{Lehmann2006}. Thus, unless $\sigma^2_X = \sigma^2_Y$ or $p = \frac{1}{2}$, the permutation distribution fails to match the unconditional distribution. 

The goal here is thus to provide a general result combining the delta method with the results of \citet{Chung2013}. First, we note that the finite forest centered MSE is equal to the original difference rescaled by $g'(\tilde{R}_B(\bm{x})) = g'(\E RF_B(\bm{x})) + o_P(1)$, so that
\[
\sqrt{B}\left[MSE_{RF}(\bm{x}; y) - \E MSE_{RF}(\bm{x};y) \right] = \sqrt{B}g'(\E RF_B(\bm{x})) \left[ RF_B(\bm{x}) - \E RF_B(\bm{x}) \right] + o_P(1)
\]
and therefore the MSE at a single point satisfies \autoref{eqn:asymp_linear} for 
\[
\begin{split}
\psi(T(\bm{x})) &= g'(\E RF_B(\bm{x}))\left[ T(\bm{x}) - \E RF_B(\bm{x}) \right] \\
\psi^\pi(T^\pi(\bm{x})) &= g'(\E RF^\pi_B(\bm{x}))\left[ T^\pi(\bm{x}) - \E RF^\pi_B(\bm{x}) \right].
\end{split}
\]
Thus, the single point MSE satisfies the conditions needed to apply Theorem 2.1 of \citet{Chung2013}. The calculation of the permutation distribution variance follows immediately from \autoref{eqn:tausqd}; the permutation distribution of the statistic  $\sqrt{2B}\left[MSE_{RF}(\bm{x}; y) - MSE_{RF^\pi}(\bm{x}; y) \right]$ converges to a normal distribution with mean 0 and variance
\[
\tau^2 = \frac{1}{1/4}\left[\frac{1}{2}\Var( g'(\E RF_B(\bm{x}))T(\bm{x})) + \frac{1}{2}\Var( g'(\E RF^\pi_B(\bm{x}))T^\pi(\bm{x})) \right].
\]
This is double the variance of \autoref{eqn:uncond}, but that in that case the variance was calculated for a $\sqrt{B}$ rescaling, and so the conditional and unconditional variances agree. Because the forest sizes used in \autoref{alg:alg1} are assumed to be the same, $p = \frac{1}{2}$, so that the permutation test for equivalence of forest predictions is automatically valid in the sense of matching the permutation and unconditional distributions. 
This argument is formalized in the following result.
\begin{theorem} \label{thm:singlepointMSE}
Let $T_{1,k_n}, ..., T_{B,k_n}$ and $T^\pi_{1, k_n},..., T^\pi_{B, k_n}$ be two collections of trees satisfying the conditions of \autoref{lem:lemma1} and \autoref{lem:centered}, and fix a test point with location $\bm{X}$ and response $Y$. Consider a test of the null hypothesis
\[
H_0: \E\left[ MSE_{RF}(\bm{X}; Y) \big| \ \bm{X}, Y\right] = \E\left[ MSE_{RF^\pi}(\bm{X}; Y) \big| \ \bm{X}, Y\right]
\]
using the statistic $\hat{\Delta} = MSE_{RF}(\bm{X}; Y) - MSE_{RF^\pi}(\bm{X}; Y)$. Then under $H_0$, the permutation distribution of $\sqrt{B}\hat{\Delta}$ converges to a normal distribution with mean 0 and variance
\[
\tau^2 = g'(\E RF_B(\bm{x}))^2\sigma^2 + g'(\E RF^\pi_B(\bm{x}))^2\sigma^2_\pi
\]
which is also the variance of the unconditional distribution of $\sqrt{B}\hat{\Delta}$, as $n\to\infty$. Thus, the permutation test attains the asymptotic Type I error rate.
\end{theorem}
\begin{proof}
The only claim that remains to be verified is that the permutation test attains the Type I error rate. Let $\Phi(\cdot)$ be the standard normal cdf, and let $\hat{J}_B(t)$ be the (random) cdf of the permutation distribution, with corresponding quantile function $\hat{J}^{-1}_B(q)$. By the argument preceding the theorem statement, we have that $\sup_t | \hat{J}_B(t) - \Phi(t/\tau)| \stackrel{p}{\to}0$. Then, by Lemma 11.2.1 of \citet{Lehmann2006}, for any number $q \in (0, 1)$,  $\hat{J}^{-1}_B(q) \stackrel{p}{\to} \tau\Phi^{-1}(q)$. In particular, for a given significance level $\alpha$, the 1-sided permutation test of $H_0$ at the level $\alpha$ has a critical value $\hat{J}_B^{-1}(1-\alpha)$ which converges in probability to $\tau \Phi^{-1}(1-\alpha)$. Thus, as $B\to\infty$,
\[
P(\sqrt{B}\hat{\Delta} \geq \hat{J}_B^{-1}(1-\alpha) | H_0) \to P(\sqrt{B}\hat{\Delta} \geq \tau \Phi^{-1}(1-\alpha) | H_0) \to \alpha .
\]\end{proof}

We now must extend this result to multipoint test sets. However, Theorem 2.1 of \citet{Chung2013} deals only with the scalar case. As such, recall that the multipoint MSE can be broken down into a sum of iid components. In particular, letting $\mathcal{T}$ be a test set consisting of $N_t$ points, it was shown in \autoref{subsec:mse} that 
\[
\sqrt{B}\left[MSE_{RF}(\mathcal{T}) - \E MSE_{RF}(\mathcal{T}) \right] = \frac{1}{\sqrt{B}} \sum_{i=1}^B \bar{T}_i + o_P(1)
\]

where $\bar{T}_i$ is an iid sequence of random variables, each with mean 0 and variance presented in \autoref{eqn:sigmabar}. Thus, the scaled and centered MSE satisfies the linearity condition presented in \autoref{eqn:asymp_linear}. In particular, $\bar{T}_1, ..., \bar{T}_B \iid P$ and $\bar{T}_1^\pi, ..., \bar{T}_B^\pi \iid P^\pi$, and we are testing $H_0: \E \bar{T}_i = \E \bar{T}_i^\pi$. Thus, because each is calculated with $B$ trees, the same results hold and the test is asymptotically valid at multiple test points. This leads naturally to the following culminating theorem, the proof of which follows an identical argument to that of \autoref{thm:singlepointMSE}, and is therefore omitted. 
\begin{theorem} \label{thm:multipointMSE}
Let $T_{1,k_n}, ..., T_{B,k_n}$ and $T^\pi_{1, k_n},..., T^\pi_{B, k_n}$ be two collections of trees satisfying the conditions of \autoref{lem:lemma1} and \autoref{lem:centered}, and fix a collection of test points $\mathcal{T}$. Consider a test of the null hypothesis
\[
H_0: \E\left[ MSE_{RF}(\mathcal{T})\ \big| \ \mathcal{T}\ \right] = \E\left[ MSE_{RF^\pi}(\mathcal{T}) \ \big| \ \mathcal{T}\ \right]
\]
using the statistic $\hat{\Delta} = MSE_{RF}(\mathcal{T}) - MSE_{RF^\pi}(\mathcal{T})$. Then, assuming $H_0$, the permutation distribution of $\sqrt{B}\hat{\Delta}$ converges to a normal distribution with mean 0 and variance given by  \autoref{eqn:sigmabar} which is also the variance of the unconditional distribution of $\sqrt{B}\hat{\Delta}$, as $n\to\infty$. Thus, the permutation test attains the asymptotic Type I error rate.
\end{theorem}

\subsubsection{Beyond the iid Approximation}
We note that the conditions of \autoref{lem:lemma1} are likely far stronger than needed to attain the result in \autoref{thm:multipointMSE}. The proofs of validity for the permutation tests rely on projecting the random forest (which is a correlated sum $\frac{1}{B} \sum_{i=1}^B T_i(\bm{x})$) onto a sum of iid random variables, $\sum_{i=1}^n \psi_n(Z_i)$ for some function $\psi_n$, to which a central limit theorem can then apply. Indeed, this is exactly the approach of the H{\' a}jek projection and H-decomposition used respectively by \citet{Mentch2016} and \citet{Wager2018}. In these works, it is roughly shown that, under constraints on the forest construction, the random forest prediction at a point $\bm{x}$ satisfies
\[
\frac{1}{\sqrt{B}}\sum_{i=1}^B\left[T_i(\bm{x}) - \E RF_B(\bm{x})\right] = \sum_{i=1}^n \psi_n(Z_i) + o_P(1).
\]
For example, if the H{\' a}jek projection is used, $\psi_n(Z_i) = \sqrt{B}\E\left[ RF_B(\bm{x}) \ \big| \ T_i(\bm{x}) \right] - \E RF_B(\bm{x})$. Moreover, as mentioned in the remark following \autoref{lem:centered}, the fact that the MSE is asymptotically linear is independent of the iid approximation, and thus the MSE for these forests is also asymptotically linear.

\section{Simulations}\label{sec:sims}

We now apply our testing procedure in a number of settings with varying regression functions and covariate structures. We simulate data from four models summarized in \autoref{tab:ydistn}, with covariate structures summarized in \autoref{tab:xdistn}. For each of our simulations, we train random forests using the \texttt{randomForest} package in R \citep{Liaw2002} using the default \texttt{mtry} parameters.

\begin{table}[H]
    \centering
    \begin{tabular}{llc}
    \hline
       Model \#  &  Data Generating Model & Covariate Structure\\
       \hline
         1 &  $Y = \beta X_1 + \beta I(X_6 = 2) + \epsilon$ & M1\\
         2 &  $ Y =  \beta\sin(\pi I(X_7 = 2)X_1) + 2\beta(X_3-.05)^2 + \beta X_4 +   \beta X_2 + \epsilon$ & M1 \\
         3 & $P(Y = 1| \bm{X}) = \text{expit}\big[\beta\sum_{j=2}^{5}X_j\big]$& M2 \\
         4 &  $ Y =  RF_{\texttt{eBird}}(\bm{X}) + \epsilon$ & \texttt{eBird} \\
         \hline
    \end{tabular}
    \caption{Distributions of $Y|\bm{X}$ for each model. $\text{expit}(z) = \frac{1}{1 + e^{z} }$.} 
    \label{tab:ydistn}
\end{table}

\begin{table}[H]
    \centering
    \begin{tabular}{l|l}
    \hline
       Model \#  &  Covariate Structure \\
       \hline
         M1 & $X_1,..., X_5 \iid Unif(0,1)$, $X_6,..., X_{10} \iid \text{Multinomial}(1, [\frac{1}{3}, \frac{1}{3}, \frac{1}{3}]^T)$ \\
         M2 & $X_1,..., X_{500} \sim \text{AR}_1(0.15)$ \\
         \texttt{eBird} & Data from \citet{Coleman2017} - 12 variables + 2 proxy variables\\ 
         \hline
    \end{tabular}
    \caption{Distribution of $\bm{X}$ for various simulation studies.}
    \label{tab:xdistn}
\end{table}

Model 1 is a standard ANCOVA model, which is intended to include both an important discrete and continuous predictor, to demonstrate the robustness of the proposed procedure to covariate type. Here we test the importance of $(X_1,X_6, X_2, X_7)$ where $X_1, X_6$ are important, $X_1, X_2$ are continuous, and $X_6, X_7$ are categorical.  Model 2 resembles the MARS data generating model \citep{Friedman1991} commonly used in random forest studies, but with a modification to include an important discrete covariate. In both settings, we draw $n = 2000$ points from the joint distribution of $(\bm{X}, Y)$, subsample sizes of $k_n = n^{0.6} \approx 95$, and build $B = 125$ trees in each forest. Predictions were made at $N_t = 100$ test points, each drawn from the same joint distribution as the training data. Note that the null hypothesis, as defined in \autoref{eqn:hypoMSE}, is conditional on the test points used. These simulations change the null hypothesis each time, because the validation set changes. Thus, the simulations mimic the common practice of random splitting the data into a training and validation fold. 

For Models 1 and 2, we focus on a marginal signal to noise ratio, which is controlled by the parameters $\beta$ and $\sigma$. We fix $\beta = 10$ across all simulations let $\sigma = 10/j$ where $j$ takes 9 equally spaced values between 0.005 and 2.25 so that for small $k$, the signal to noise ratio (SNR) is small. The results are shown in \autoref{fig:simresults}. We see that the test maintains the nominal type I error rate and attains high power for marginal SNRs near 1 for all variables except $X_7$ in Model 2.  Note also that the type I error rate appears insensitive to the covariate structure. In the MARS model, we see that the test has more power against $X_3$ than $X_7$, because $X_7$ is only important insofar as it interacts with $X_1$. 

 Model 3 is an adaptation of the model used in \citet{Candes2016} for high-dimensional correlated data. Here we test for the significance of $X_2$, which is important, and also $X_1$ and $X_{500}$, which are unimportant, but $X_1$ is highly correlated with $X_2$ and $X_{500}$ is much more weakly correlated. \citet{Candes2016} demonstrated that the standard logistic regression p-values in this situation are far from uniform under $H_0$, so that standard parametric inference may not be valid. Random forests, on the other hand, have been shown \citep{Biau2012, Scornet2015} to be largely insensitive to the dimension of the ambient feature space, and instead sensitive only to the ``strong" feature space. This setting helps to explore the utility of our method in the high dimensional sparse signal case. 

We limit $n= 600$ so that $p/n$ is not small, though the dimension of the strong features is still small relative to $n$. We let $k_n = n^{0.6} \approx 46$, $B = 125$, $N_t = 100$, and vary the $\beta$ coefficient according to 8 equally spaced values between 0.01 and 2.5 and also for 7 equally spaced values between 5 and 20.  The results are shown in the bottom panel of \autoref{fig:simresults}.  Note that the test resolves the biased p-value issue associated with the standard glm procedure and is still able to attain reasonable power for the effect of $X_2$. The power is likely limited by the fact that for large $\beta$, the change in the marginal effect of each covariate only changes $P(Y = 1 | \bm{X})$ slightly due to the rapidly decaying first derivative of the $\text{expit}(z)$ function.

Finally, we turn to Model 4 where the true data generating model is a random forest. We utilize a dataset from \citet{Coleman2017} describing the occurrence of tree swallows and to construct $ RF_{\texttt{eBird}}$, we draw 5000 points from the data, and train $ RF_{\texttt{eBird}}$, a random forest with $\texttt{mtry} = 9$ and 1000 total trees. To simulate from this model, we draw (without replacement) samples of size $n$ from the remaining 20727 points, predict at them using $RF_{\texttt{eBird}}$, and add Gaussian noise. We test for the effect of two variables: $\texttt{eff.hours}$, which corresponds to the number of hours a user expended upon a hike, and $\texttt{dfs}$, which is a fractional measurement of day of year. We further include two proxy variables (not used to train $ RF_{\texttt{eBird}}$), defined as $\texttt{eff.hours.proxy} = \frac{\texttt{eff.hours} + Z_{0.5}}{\sqrt{\Var(\texttt{eff.hours})} + 0.5} $ and $\texttt{dfs.proxy} = \frac{\texttt{dfs} + Z_{0.025}}{\sqrt{\Var(\texttt{dfs})} + 0.025} $ where $Z_\sigma$ is a standard normal random variable with variance $\sigma^2$. The purpose of this construction is that the proxy variables' relationship with $Y$ is solely dictated by their dependence on their original copy. 
\begin{figure}
\centering
\includegraphics[width = .925\textwidth]{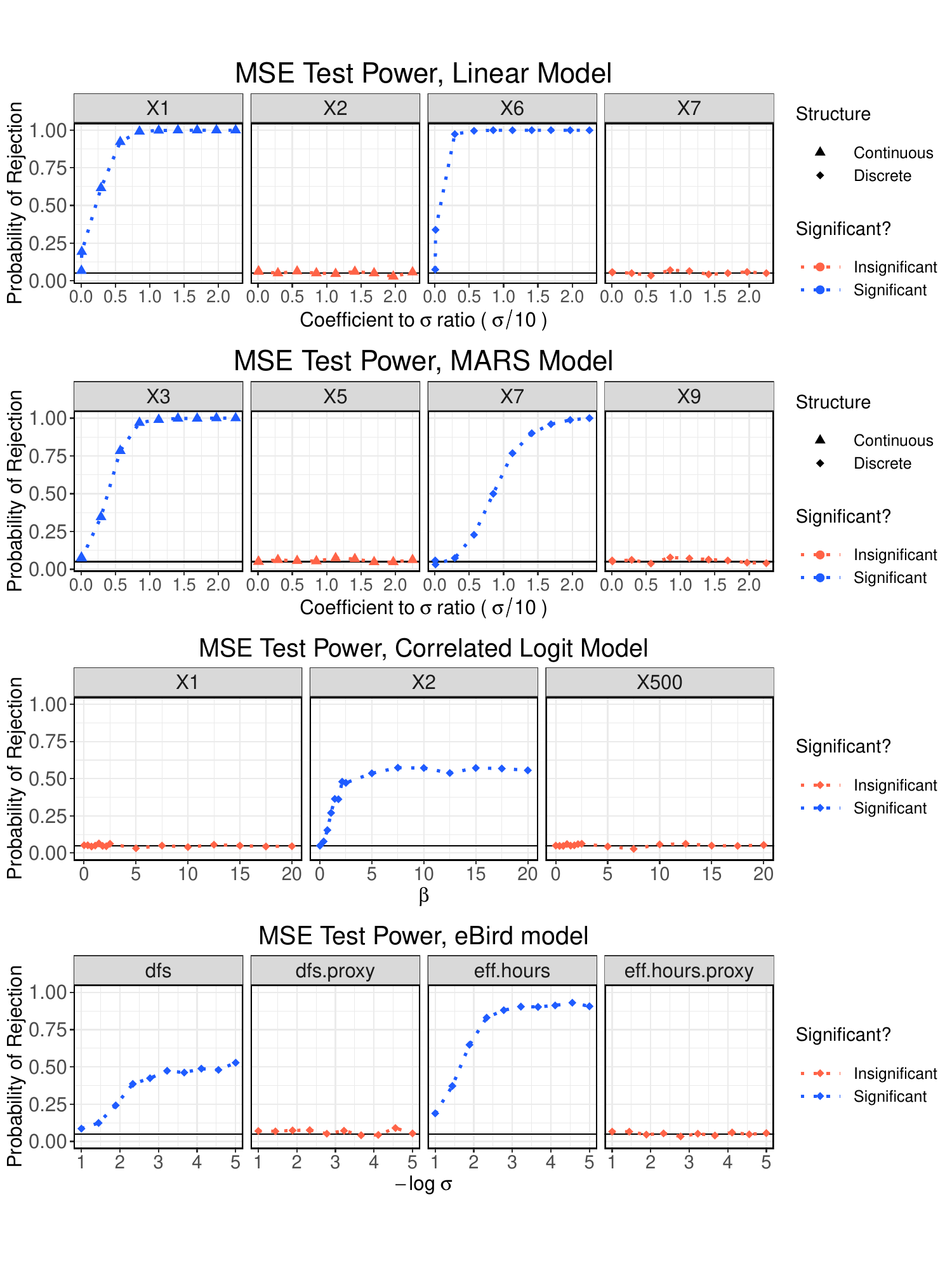}
\caption{Simulation results for each of the models from \autoref{tab:ydistn}. Black line corresponds to $\alpha = 0.05$, the nominal level} \label{fig:simresults}
\end{figure}

In Model 4, we let $n = 2000$, $k_n = n^{0.6}$, $B = 125$, $N_t = 100$, and let $\sigma = e^{-j}$ for 10 values of $j$ equally spaced between 1 and 5.  The results of this simulation are show in \autoref{fig:simresults}. We see that again the test maintains the nominal type I error rate with modest power for signal variables. Moreover, the procedure correctly identifies the true variables as important over their proxies.

\section{Applications to Ecological Data}
\label{sec:models}
We now apply our testing procedure on two ecological datasets where random forests have been shown to perform well in recent work.

\textbf{eBird:  }We first consider the eBird data described in the previous section to construct a simulated random forest model.  Here we utilize the original data as considered in \citet{Coleman2017}.  The standard task is to predict tree swallow \texttt{occurrence} during the fall migration season in a particular geographic area referred to as Bird Conservation Region (BCR) 30. This is a Citizen Science project where observers submit reports detailing when and where they recorded observations.  The response in each row of the data is either 0 or 1 corresponding to whether a tree swallow was observed during that particular outing.  Features include information about latitude, longitude, time of year, as well as observer, environmental, temperature, and land cover characteristics. The data consists of $n = 25727$ observations on $23$ features, gathered between 2008 and 2013. 
\citet{Coleman2017} carry out a testing procedure based on the parametric approach in \citet{Mentch2016} but due to the limitations described in previous sections, are limited to a test sample of only 25 points.

We first apply \autoref{alg:alg1} to test the importance of any variables in predicting \texttt{occurrence}, analogous to an overall F-test in multiple linear regression. Here we select 15\% of the available observations ($\approx$ 3800 points) uniformly at random to serve as the test set where the hypotheses will be evaluated.  The random forests were trained with the \texttt{ranger} package using the default $\texttt{mtry} = 4$, subsamples of size $k_n = n^{0.6}$, and consisting of $B = 250$ trees in each.  The results are shown on the left hand side of  \autoref{fig:eBird_data}. There is clear evidence for signal in the data, with an overall p-value of $p  < 0.0001$.  Next, to produce an output similar to the out-of-bag importance scores traditionally computed, we repeat the testing procedure for each covariate individually, recording the marginal importance for each as the number of standard deviations away that the original MSE difference is from the center of the permutation distribution.  The results are shown in the right hand side of \autoref{fig:eBird_data}. We see that \texttt{dfs}, which corresponds to the day of the year, \texttt{eff.hours}, which corresponds to a users' effort (in time), and \texttt{aster.elev}, which corresponds to elevation, are the most important features. Time of year (\texttt{dfs}) and elevation (\texttt{aster.elev}) have an intuitive relationship with occurrence, serving as proxies for climate conditions. Larger \texttt{eff.hours} suggest that a user spent more time out in the field on a particular day, meaning they were more likely to observe a tree swallow because of increased birding time.

\begin{figure}
    \centering
        \begin{subfigure}[b]{0.475\textwidth}
        \centering
	\includegraphics[width = \textwidth]{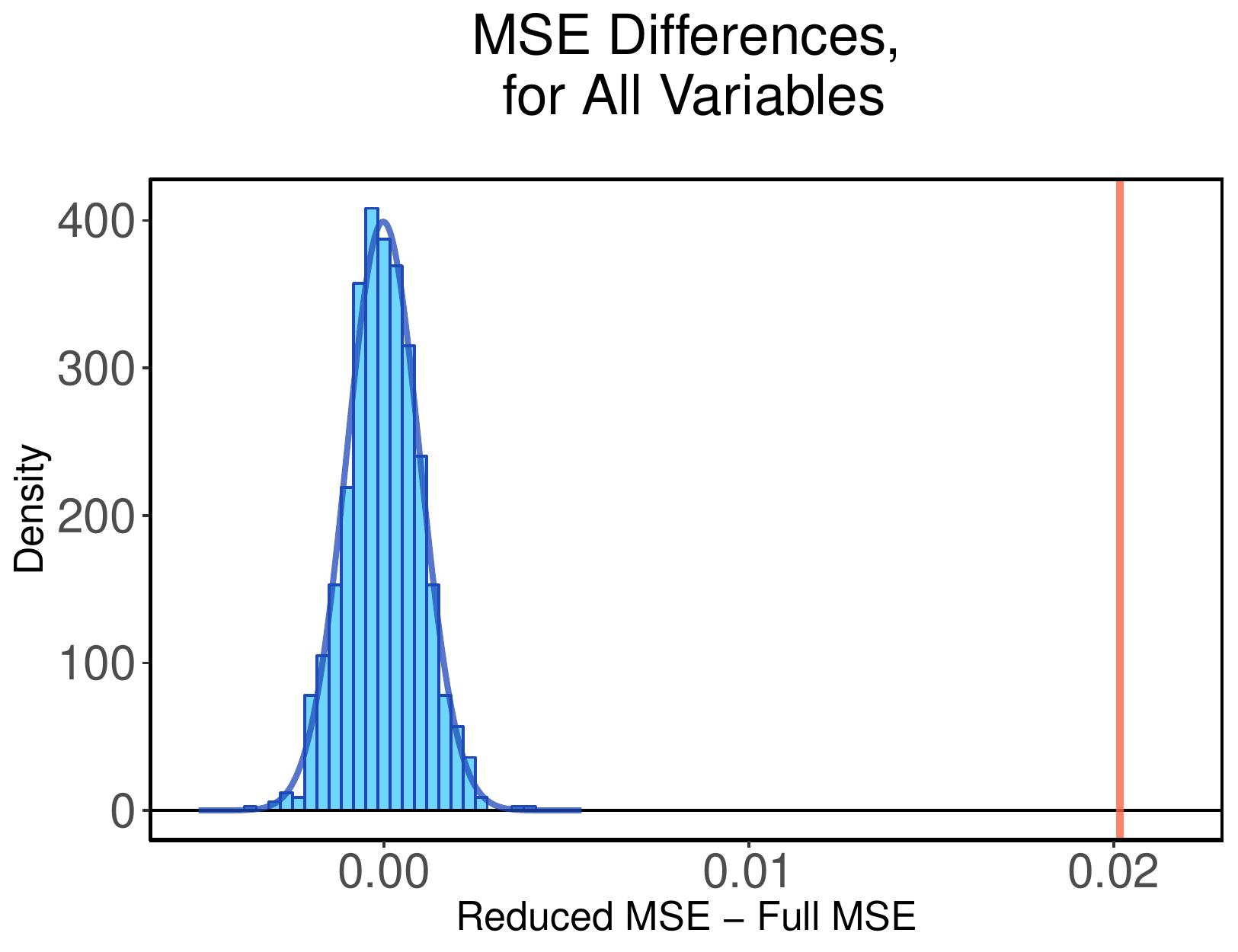}
	\caption{MSE Test for all covariates}
    \end{subfigure}
    \begin{subfigure}[b]{0.475\textwidth}
        \centering
	\includegraphics[width = \textwidth]{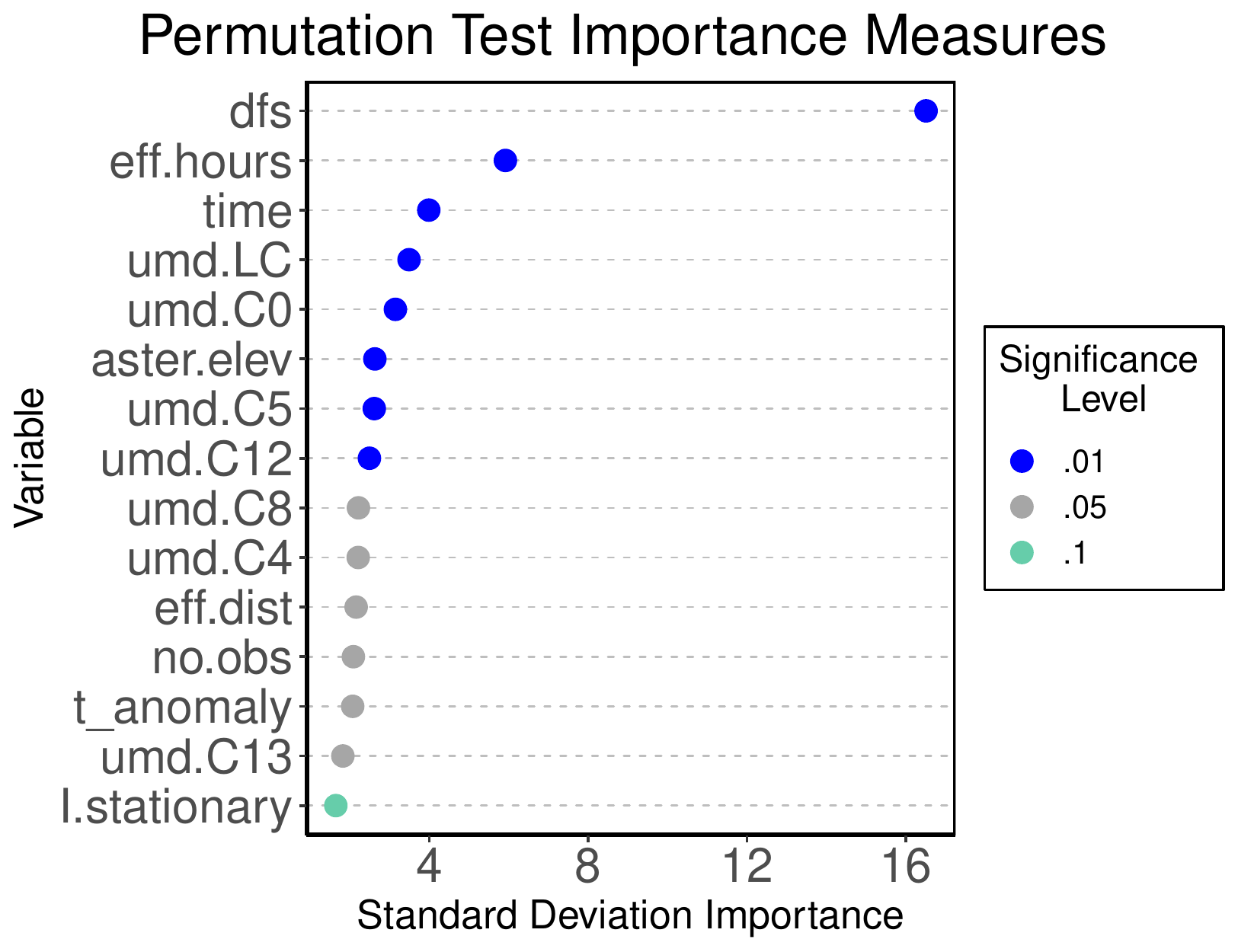}
	\caption{Marginal Importance Scores}
    \end{subfigure}
	\caption{Results on the eBird data from \citep{Sullivan2009, Sullivan2014}. Red line indicates observed value, and histograms of differences in MSE after permutation are overlayed by an estimated normal density.}\label{fig:eBird_data}
\end{figure}

\textbf{Forest Fires:  }\citet{Cortez2007} sought to predict \texttt{$\log(1 + \text{area})$} burned by several fires in northern Portugal using covariate information on location, time of year, and local weather characteristics. The data contains $n = 537$ observations on $13$ features. \citet{Cortez2007} found that a naive mean predictor attained the lowest RMSE - suggesting that there is weak signal in the data.  We carry out our testing procedure in exactly the same fashion as the eBird data, using $\texttt{mtry} = 12$ and $k_n = n^{0.6} \approx 43$, $B = 250$ trees for the importance test and $B = 500$ trees for the overall test; results are shown in \autoref{fig:ff_data}.  The overall test suggests that there is signal in the data ($p = 0.0040$), albeit a weaker effect than in the preceding eBird case study.  The importance procedure suggests that only \texttt{wind} -- the wind speed at the location of the fire -- is significant at the 0.05 level.

\begin{figure}
    \centering
    \begin{subfigure}[b]{0.475\textwidth}
    	\centering
	\includegraphics[width = \textwidth]{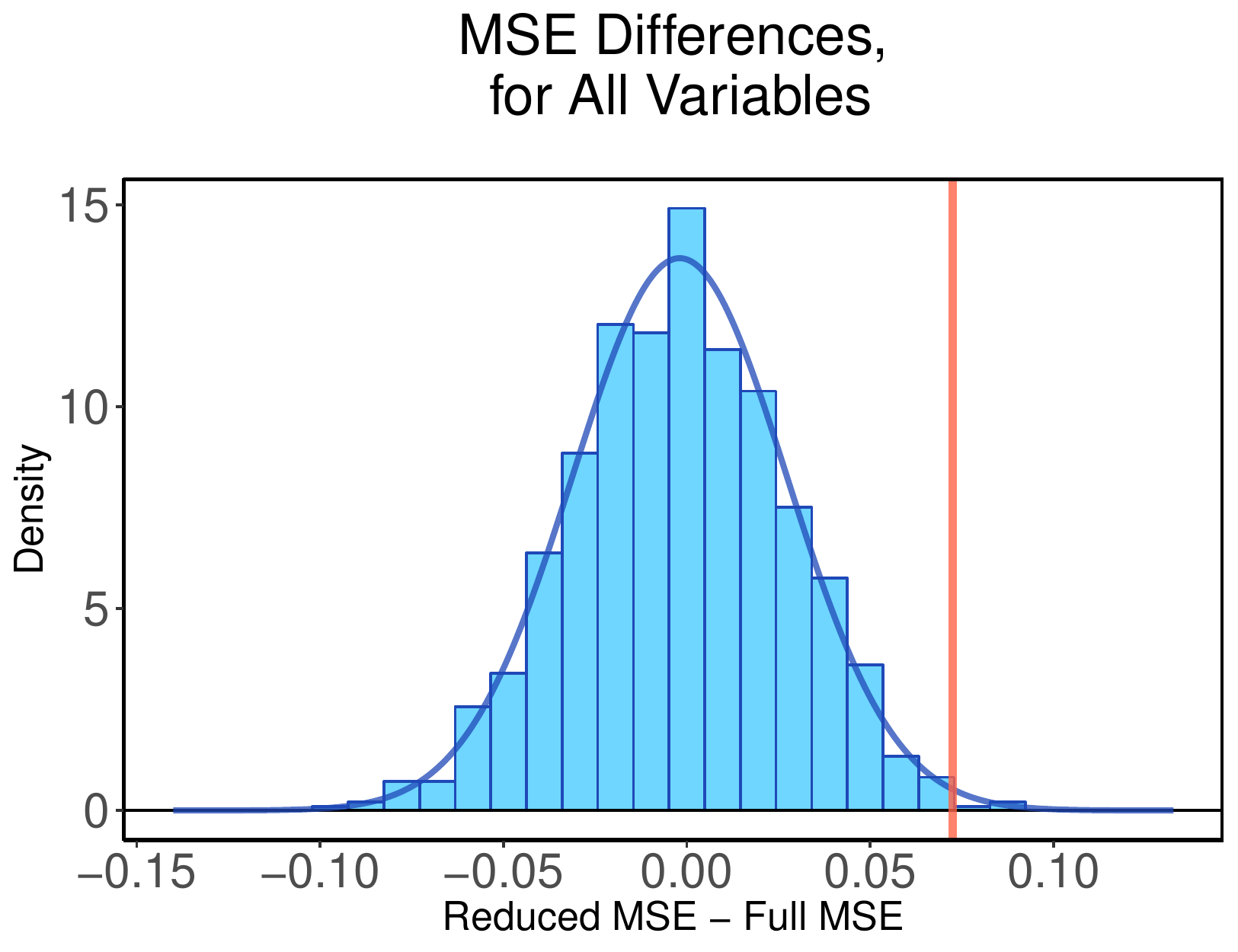}
	\caption{MSE Test for all covariates}
    \end{subfigure}
    \begin{subfigure}[b]{0.475\textwidth}
        \centering
	\includegraphics[width = \textwidth]{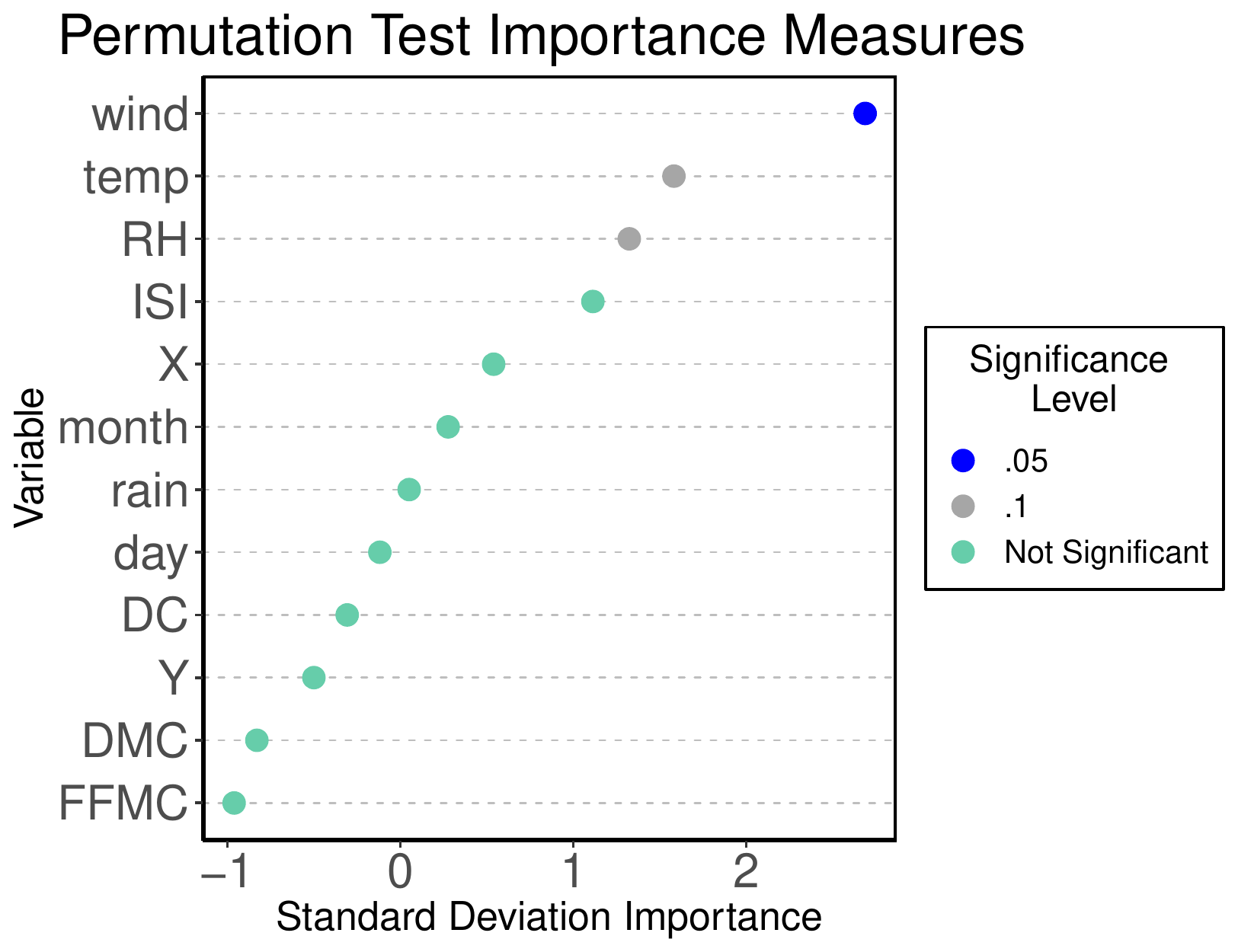}
	\caption{Marginal Importance Scores}
    \end{subfigure}
	\caption{Results on the forest fire data from \citet{Cortez2007}. Red line indicates observed value, and histograms of differences in MSE after permutation are overlayed by an estimated normal density.} 					\label{fig:ff_data}
\end{figure}

\section{Discussion} \label{sec:discus}
The work here presents a formal hypothesis testing framework for evaluating the predictive significance of covariates in a random forests model which, unlike existing approaches, is both computationally efficient and statistically valid, placing hypothesis tests with random forests firmly within the grasp of applied researchers.  Previously suggested parametric approaches are computationally prohibitive and place severe restrictions on where the hypotheses can be evaluated while the popular heuristic out-of-bag (oob) approaches are easily computed but also easily fooled by correlated and/or categorical covariates.  We note further that while the ensemble nature of random forests presents a natural context for such tests, much of the theoretical backing for this procedure is largely agnostic to the particular class of base-learner models being constructed.  

Besides its feasibility, this permutation approach also offers some flexibility in the kinds of problems open to investigation by practitioners.  Consider, for example, the mediator detection problem arising frequently in medical studies wherein a covariate $X_1$ is a \emph{mediator} for another covariate $X_2$ whenever the effect of $X_2$ on the response is nullified (or substantially lessened) by including $X_1$ in the model.  The same two-step process often employed with linear models can be carried out with random forests using the tests developed here:  first determine whether $X_2$ is significant without $X_1$ in the model, then test whether the significance of $X_2$ disappears whenever $X_1$ is included.  

One potential criticism of the approach presented here may be that it becomes more computationally burdensome whenever one wishes to evaluate the significance of all available covariates one at a time.  Note however that by construction, we need only build relatively few trees to conduct each test and thus in small or even moderate dimensions, simply repeating our permutation approach $p$ times is still more computationally efficient than carrying out even a single parametric test.Finally, we note that the validity of our approach was verified by arguing that the random forest trees behave like an iid sequence asymptotically. As argued in \autoref{subsec:permtests}, random forest predictions are often near an iid sum, so that the linearity condition of \citet{Chung2013} may be applicable in many wider cases.  

\bigskip
\begin{center}
{\large\bf SUPPLEMENTARY MATERIAL} 
\end{center}
R code for implementing the testing procedure as well as all simulation examples is provided here and publicly available.
 
\bigskip
\begin{center}
{\large\bf ACKNOWLEDGEMENTS}
\end{center}
This research was supported in part by the University of Pittsburgh Center for Research Computing through computing resources. We specifically acknowledge the assistance of Kim Wong.  LM was supported in part by NSF DMS-1712041. We thank the thousands of eBird participants for their contributions and support for eBird from The Wolf Creek Foundation and the National Science Foundation (ABI sustaining: DBI-1356308; computing support from CNS-1059284). We also thank Giles Hooker for useful feedback.

\bigskip

\bibliographystyle{Chicago}

\appendix

\section{Proofs of Technical Results} \label{appdx:proofs}

We now provide the technical details and proofs for theoretical discussion in Section 3.  For completeness, theorems and lemmas are restated. 

\textbf{Theorem 1.} \textit{Under the exchangeability conditions outlined in Section 3.1, denote a sequence of (potentially randomized) trees trained on subsamples from $\df_n$ as $\{T_k(\cdot)\}_1^\infty$. Moreover, consider an independently drawn test point, $\bm{Z}^* = (\bm{X}^*, Y^*) \sim F$. Then, the residuals}
\[
r_k = T_k(\bm{X}^*) - Y^*
\]
\emph{form an infinitely exchangeable sequence of random variables.}

\begin{proof}
Let $\Xi$ be the distribution of randomization parameters, and let $\mathcal{S}_{k_n}(\df_n)$ be the distribution of subsamples of size $k_n$ drawn uniformly from the original data. Then, to construct a tree, we have the following procedure:
\begin{enumerate}
\item Draw $\df_{k_n}^* \sim \mathcal{S}_{k_n}(\df_n)$
\item Draw $\xi \sim \Xi$
\item Draw $\bm{Z}^* \sim F$
\item Construct a tree according to some combining function, say $\phi$ , of $\xi, \df_{k_n}^*$, i.e. $T = \phi(\xi, \df_{k_n}^*)$.
\end{enumerate}
Each draw is done independent of the other draws. Repeating (1) and (2) independently gives iid sequences $\{\df_{l, k_n}^*)\}_{l = 1}^\infty$ and $\{\xi_j\}_{j=1}^\infty$.  Then, the sequence
\[
T_1 = \phi(\xi_1, \df_{1, k_n}^*),\ T_2 =\phi(\xi_2, \df_{2, k_n}^*), ...
\]
is a mixture of iid sequences, where the mixture is directed (in the sense of \citet{Aldous1985})  by $\df_n$.  So, $\{T_l \ |\ \df_n\}$ is exactly an iid sequence of functions. Further,
$\{r_l \ |\ \df_n, \bm{Z}^*\}$ is an iid sequence of random variables, and thus the conclusion follows from the converse of DeFinetti's Theorem.
\end{proof}

See \citet{Aldous1985} page 29 for more details on this construction.

We turn now to Lemma 1 from Section 3.2, which establishes asymptotic pairwise independence.

\textbf{Lemma 1.} \textit{Consider a collection of $B_n$ trees  built from a training dataset of size $n$ on subsamples of size $k_n$, say $\{T_{j, k_n}\}_{j=1}^{B_n}$, satisfying \autoref{cond:cond1}. Then, as long as $k_n/\sqrt{n} \to 0$ and}
\[
 \binom{B_n}{2} \log\bigg[\frac{\binom{n- k_n}{k_n}}{\binom{n}{k_n}}\bigg] \to 0
\]
\emph{the infinite sample sequence of trees, $\{T_{1,\infty, k_\infty}, ..., T_{B,\infty, k_\infty},...\}$ is an infinite sequence of pairwise independent random functions.}
 
\begin{proof}
\autoref{cond:cond1} guarantees the existence of a limiting random variable. 

It is sufficient to show that asymptotically, the trees are trained using independent training samples, because we have assumed that our original data are iid. Define the indices of a subsample  in the following way:
\[\ind(\df^*_{k_n}) := \{j \in \{1,..., n\}: Z_j \in \df^*_{k_n} \} . \]
Then, by the assumption that the $Z_k$ are independent,
$$\df^*_{k_n,j}\ \indep\ \df^*_{k_n,l} \iff |\ind(\df^*_{k_n,j}) \cap \ind(\df^*_{k_n,l})| = 0$$
so, it is sufficient to show that
\[
\lim_{n\to\infty}P(|\ind(\df^*_{k_n,j}) \cap \ind(\df^*_{k_n,l})| = 0) = 1, \ \forall \ j \neq l .
\]
Note that if $k_n \geq n/2$, this event has probability 0, so choose $n$ so that $n > 2k_n$. Then
\begin{align*}
P(|\ind(\df^*_{k_n,j}) \cap \ind(\df^*_{k_n,l})| = 0) &= \frac{\binom{n- k_n}{k_n}}{\binom{n}{k_n}} \\
&= \frac{((n-k_n)!)^2}{n!(n-2k_n)!} \\
&= \frac{(n-k_n)!}{n!} \times \frac{(n-k_n)!}{(n-2k_n)!} \\
&= \frac{(n-k_n)(n-k_n-1)...(n-2k_n+1)}{n(n-1)...(n-k_n+1)}.
\end{align*}
There are $k_n$ terms in both the numerator and denominator here, so we can separate the product in the term above as
\begin{align*}
P(|\ind(\df^*_{k_n,j}) \cap \ind(\df^*_{k_n,l})| = 0) &= \frac{n-k_n}{n} \times \frac{n-k_n-1}{n-1}\times ... \times \frac{n-2k_n+1}{n-k_n+1} . \\
&\geq \left(\frac{n-2k_n+1}{n}\right)^{k_n} \\
&= \left( 1 - \frac{2k_n + 1}{n}\right)^{k_n} \\
&= \exp\left[k_n \log\left(1 - \frac{2k_n + 1}{n}\right) \right] \\
&\approx \exp\left[k_n\left(-\frac{2k_n + 1}{n}\right) - \frac{k_n}{2}\left(\frac{2k_n + 1}{n} \right)^2 \right] \\
&\approx \exp\left[-\frac{2k_n^2 + k_n}{n}\right] \\
&\approx 1 
\end{align*}
where $a_n \approx b_n$ means that $\lim_{n\to\infty} a_n/b_n = 1$, and we have used the Taylor expansion of $\log(1-x)$ in the above.

This means that two pre-specified subsamples will be independent in the limit. Next, we need to ensure that this holds for all subsamples, i.e. 
\[
P\bigg( \bigcap_{j \neq l}\{ |\ind(\df^*_{k_n,j}) \cap \ind(\df^*_{k_n,l})| = 0\}\bigg) \to 1 .
\]

For $B_n$ trees, there are $\binom{B_n}{2}$ subsample pairings, each drawn independently. Thus
\begin{align*}
P\bigg( \bigcap_{j \neq l} \{ |\ind(\df^*_{k_n,j}) \cap \ind(\df^*_{k_n,l})| = 0\}\bigg) &= \prod_{j\neq l} P( |\ind(\df^*_{k_n,j}) \cap \ind(\df^*_{k_n,l})| = 0)  \\
&= \bigg(\frac{\binom{n- k_n}{k_n}}{\binom{n}{k_n}}\bigg)^{\binom{B_n}{2}}.
\end{align*}
Next, by assumption,
\[
 \log P\bigg( \bigcap_{j \neq l} \{ |\ind(\df^*_{k_n,j}) \cap \ind(\df^*_{k_n,l})| = 0\}\bigg) = \binom{B_n}{2} \log\bigg[\frac{\binom{n- k_n}{k_n}}{\binom{n}{k_n}}\bigg]  \to 0
\]

so that the probability of this event goes to 1.
\end{proof}

After \autoref{lem:lemma1}, we next need to prove \autoref{lem:centered}, whose purpose is to show that the observed MSE is asymptotically centered around its own expectation.
 
 \textbf{Lemma 3}
\textit{Assume the conditions needed from \autoref{cor:CLT}. Additionally, assume that $g$ has at least $k$ derivatives for some $k \geq 3$ , and that $g^{(k)}(x) < \infty$ for all $x$. Further, assume that $\E | T_i(\bm{x})|^k <\infty$. Then, }
\[
\sqrt{B}\left[\E g(RF_B(\bm{x}) - g(\E RF_B(\bm{x})) \right]= \frac{g''(\E RF_B(\bm{x}))\sigma^2}{2\sqrt{B}} + o(B^{-3/2}).
\]
 \begin{proof} We rely on a result presented in \citet{Oehlert1992}, which states that under the conditions presented in the lemma statement, 
\begin{equation} \label{eqn:Oehlert}
\E g(RF_B(\bm{x})) = g(\E RF_B(\bm{x}))  + \frac{g''(\E RF_B(\bm{x})) \sigma^2}{2B} + o(B^{-2}).
\end{equation}
Thus, the result follows from multiplying both sides of \autoref{eqn:Oehlert} by $\sqrt{B}$ and rearranging terms.
 \end{proof}

Next, we move on to the proof of \autoref{prop:prop1} from \autoref{subsec:treespec}, which gives that the trees typically utilized in a random forest obey the necessary regularity conditions for \autoref{cor:CLT}.

\textbf{Proposition 1.} 
\textit{Assume that $Y = m(\bm{X}) + \epsilon$, where $m(\cdot)$ is continuous on the unit cube. Let $\mathcal{X} = [0,1]^p$, and assume that $X_{i,j} \iid Unif(0,1)$ for $i = 1,..., n$ and $j = 1,...,p$. Then, let $T_n(\bm{x})$ be a tree trained on iid pairs $(\bm{X}_1, Y_1),..., (\bm{X}_n, Y_n)$ such that each leaf of the tree contains a single observation. Further, assume the trees satisfy the following two conditions:}

\begin{enumerate}[(i)]
\item \textit{$\exists \gamma > 0$ such that $P(\text{variable $j$ is split on}) > \gamma$ for $j \in \{1,...,p\}$}
\item \textit{Each split leaves at least $\gamma n$ observations in each node.}
\end{enumerate}

\textit{Then, for each $\bm{x}  \in \mathcal{X}$
}\[
T_n(\bm{x}) \stackrel{d}{\to} Y|\bm{X} = \bm{x} \ \text{as} \ n \to \infty
\]

\begin{proof}
Each tree divides $\mathcal{X}$ into a partition of rectangular subspaces, corresponding to leaves of the tree. Following \citet{Meinshausen2006}, for each point $\bm{x}$ (with coordinates $ [x_1,..., x_p]$), let $\ell(\bm{x})$ denote the unique leaf of the tree containing $\bm{x}$. Let $R_\ell(\bm{x})$ be the rectangular subspace of $[0,1]^p$ corresponding to a particular leaf $\ell(\bm{x})$. The rectangular nature of the subspaces means that for each input feature, $R_\ell$ can be expressed as
\[
R_\ell(\bm{x}) = \bigotimes_{i = 1}^p [a(\bm{x}, i), b(\bm{x}, i)]
\]
where $0 \leq a(\bm{x}, i) \leq x_i \leq b(\bm{x}, i) \leq 1$ are scalars inducing an interval in dimension $i$. Then, the tree (by the existence of the requisite $\gamma$) satisfies the conditions of Lemma 2 in \citet{Meinshausen2006}, so that $\max_i |a(\bm{x}, i) - b(\bm{x}, i)| \stackrel{p}{\to} 0$. Let $\bm{a}(\bm{x}) = [a(\bm{x}, 1),..., a(\bm{x},p)]$ and similarly define $\bm{b}(\bm{x})$, so that the previous sentence implies: $\bm{a}(\bm{x}) \stackrel{p}{\to} \bm{b}(\bm{x})$. We therefore also see that $a(\bm{x}, i), b(\bm{x} ,i) \stackrel{p}{\to} x_i$ for all $i$. 

The trees are fully grown, so the tree prediction at the point $\bm{x}$ is given by
\[
T_n(\bm{x}) = \sum_{k=1}^n I(\bm{X}_k \in R_\ell(\bm{x}))Y_k
\]
i.e. the response for the observation whose leaf contains $\bm{x}$. As such, let $k^*$ be the index corresponding to the observation who shares a leaf with $\bm{x}$, so that $T_n(\bm{x}) = Y_{k^*}$. We can deconstruct the event $\bm{X}_{k^*} \in R_\ell(\bm{x})$ as
\begin{align*}
\{\bm{X}_{k^*} \in R_\ell(\bm{x}) \} &= \bigcap_{i=1}^p \{ a(\bm{x}, i) \leq X_{i, k^*} \leq b(\bm{x}, i) \} .
\end{align*}
Thus, in the limit,  $a(\bm{x}, i), b(\bm{x} ,i) \stackrel{p}{\to} X_{i,k^*}$, and so $X_{i,k^*}  \stackrel{p}{\to} x_i$ for all $i$. Further, continuity of $m$ yields that $m(\bm{X}_{k^*})  \stackrel{p}{\to} m(\bm{x})$. Thus, we see that, in the limit
\[
Y_{k^*} = m(\bm{x}) + \epsilon_{k^*} \stackrel{d}{=} m(\bm{x}) + \epsilon \stackrel{d}{=} Y| \bm{X} = \bm{x}
\]
because $\epsilon_{k^*}$ is independent of the location of $\bm{X}$.
\end{proof}

\section{Additional Simulations} \label{appdx:more_sims}
We include some additional simulations here to demonstrate the following points.
\begin{enumerate}
\item The accuracy of the permutation distribution approximation of the Gaussian.  The theory outlined in \autoref{sec:theory} establishes that the difference in MSEs between forests is asymptotically Gaussian but the difficulty in estimating the resulting variance largely restricts its direct usage in practical settings.  We go on to demonstrate that the permutation distribution approaches this distribution, thereby circumventing the need for a direct variance estimate.  The simulations below present empirical evidence that this approximation is reasonable in practical settings.
\item The instability of the variance estimation procedures laid out in \citet{Wager2014a} and \citet{Mentch2016}. Clearly variance estimation is useful for developing confidence intervals about random forest predictions, which in the case of pointwise consistency (as in the honest trees proposed by \citet{Wager2018}), are also valid for the underlying regression function. However, in the hypothesis testing framework, these estimates are useful only insofar as they allow for calculation of a test statistic. These variance estimates, such as the infinitesmal jackknife of \citet{Wager2014a}, recommend building $B = \mathcal{O}(n^{\beta})$ trees where $\beta \geq 1$. We demonstrate that this recommendation cannot be violated.
\item The robustness (and potential weaknesses) of the proposed procedure to different random forest implementations. In particular, we want to study the effect of larger subsamples/more trees. The theoretical results presented in \autoref{sec:theory} rely on treating the tree predictions as iid. Clearly, this is never true in practice, and some theoretical justification for the effects of this being small were presented in \autoref{sec:discus}. 
\end{enumerate}

\subsection{Normality of Permutation Distributions}
Here we provide a concise simulation demonstrating the accuracy of the permutation distribution approximation of the Gaussian in a practical setting.  We simulate $n = 2000$ training observations from Model 2 with covariate structure M1 as described in \autoref{sec:sims}.  Specifically, we consider the model $ Y =  \beta\sin(\pi I(X_7 = 2)X_1) + 2\beta(X_3-.05)^2 + \beta X_4 +   \beta X_2 + \epsilon$ where we sample covariates according to $X_1,..., X_5 \iid Unif(0,1)$ and $X_6,..., X_{10} \iid \text{Multionimial}(1, [\frac{1}{3}, \frac{1}{3}, \frac{1}{3}]^T)$.  Here we use $\beta = 10, \sigma = 10$, along with $N_t = 100$ test observations and apply our procedure to test for the significance of $X_3$ (important) and $X_5$ (unimportant). The random forests each consist of $B = 200$ trees trained on subsamples of size $k_n = n^{0.6}$, with $\texttt{mtry} = 3$. The resulting permutation distributions are shown in \autoref{fig:perm_norm}.

\begin{figure}[H]
    \centering
    \begin{subfigure}[b]{0.475\textwidth}
        \centering
	\includegraphics[width = \textwidth]{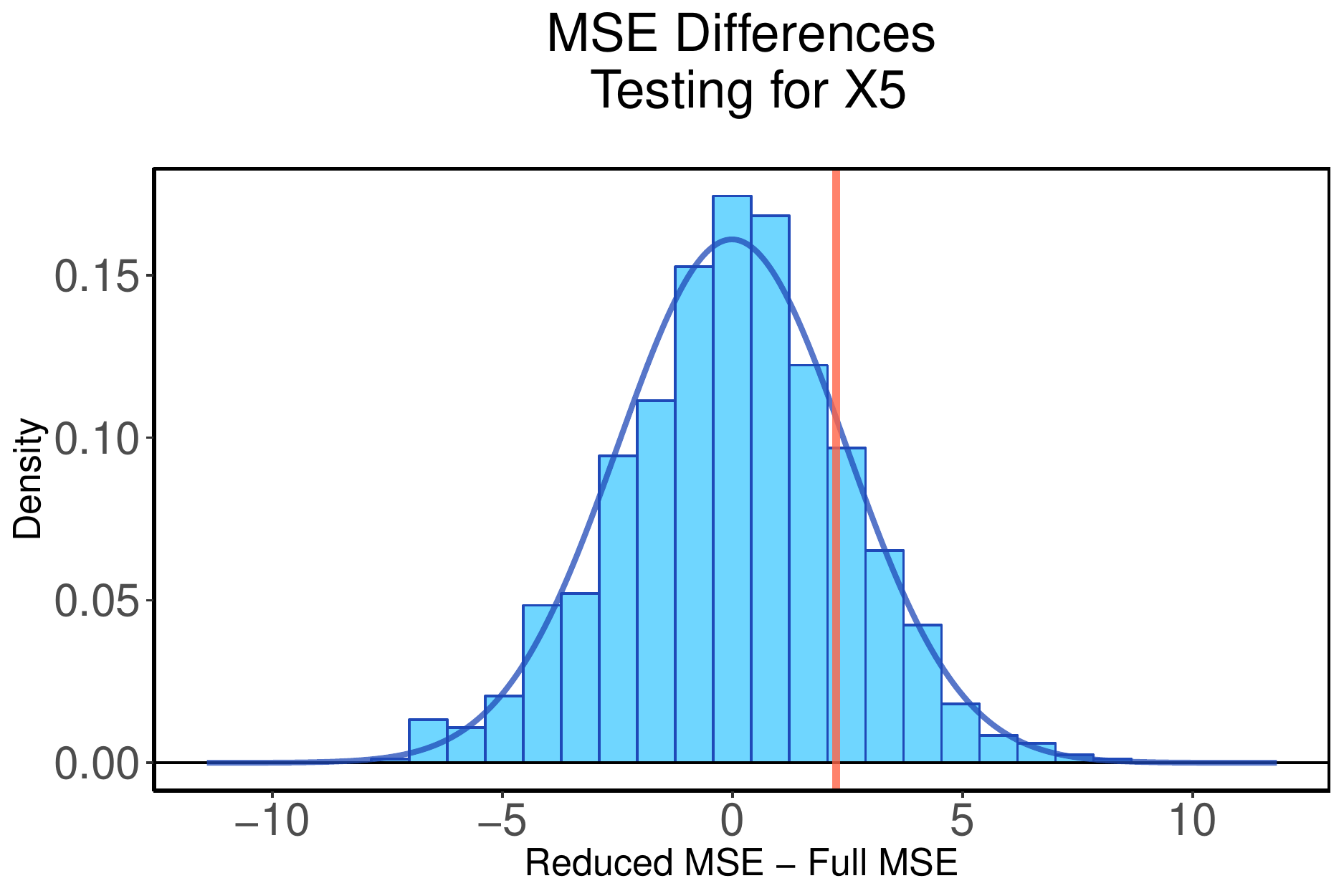}
	\caption{Under $H_0$}
    \end{subfigure}
    \begin{subfigure}[b]{0.475\textwidth}
        \centering
	\includegraphics[width = \textwidth]{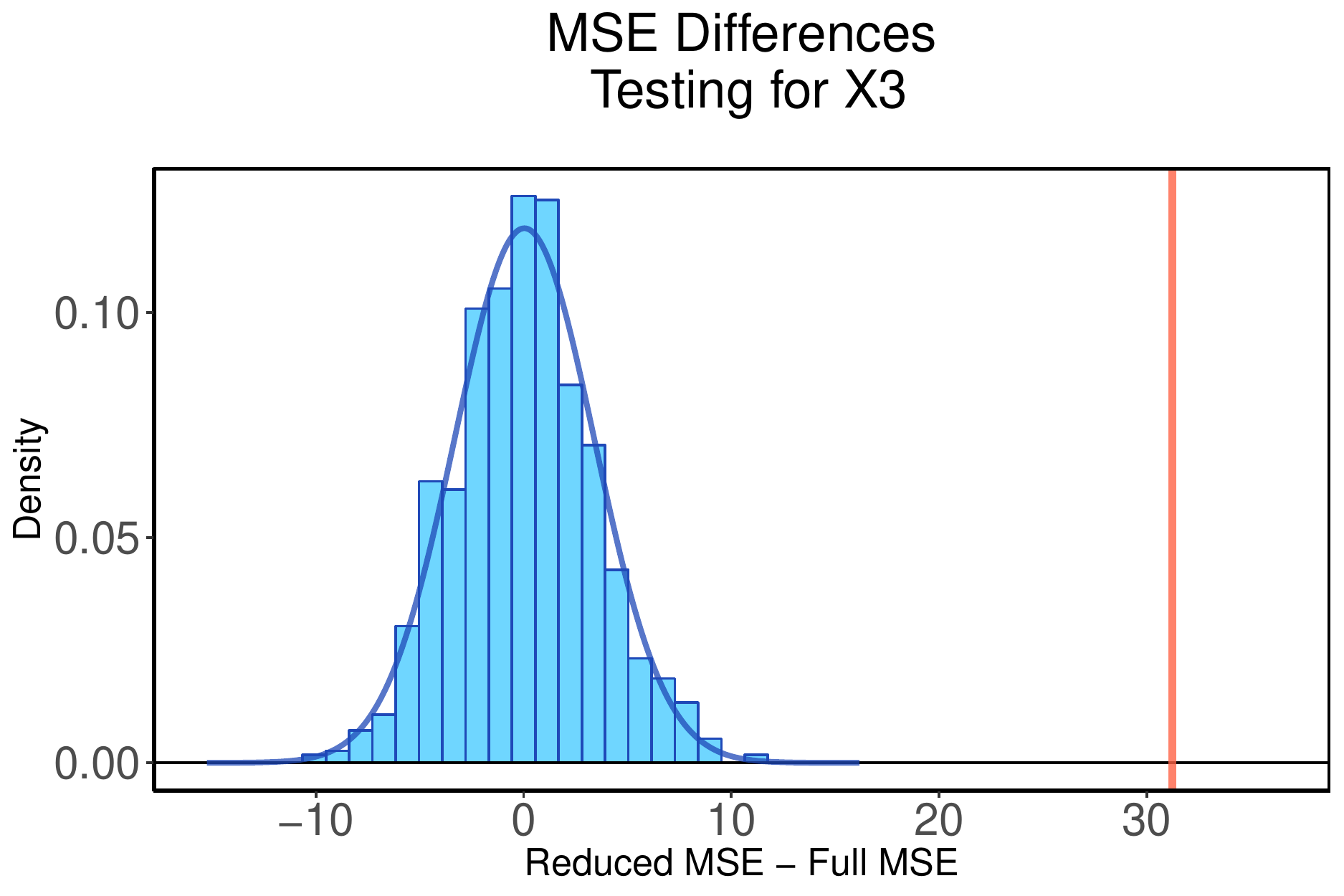}
		\caption{Under $H_1$}
    \end{subfigure}
	\caption{Permutation distributions of $\Delta_B$. Red line indicates observed value, and histograms are overlayed by an estimated normal density.} 					\label{fig:perm_norm}
\end{figure}

These plots demonstrate that the permutation distributions do approximate a Gaussian distribution. Moreover, in the null case, the observed $\Delta_B$ lies squarely in the center of the distribution, while in the alternative case, $\Delta_B$ lies far away from the center. Next, we more formally investigate the power/validity of the testing procedure.

\subsection{Variance Estimation Instability}
Here, we use the infinitesmal jackknife (IJ), as implemented in the \texttt{ranger} package \citep{Wright2015}, to estimate the variance of a random forest prediction at a given point. In particular, we simulate data from Model 2 from \autoref{tab:ydistn}, train a subsampled random forest, and record the IJ variance estimate of random forest prediction at $X_1 = ... = X_5 = 0.5$ and $X_6 = ... = X_{10} = 2$. We use $n = 2000$, $k_n = n^{0.5} \approx 44$, and vary the number of trees. Often times, the IJ variance estimate is negative, leading to a \texttt{NaN} output from the IJ software. These instances represent a case when the IJ estimate is useless to a practitioner, and as such, we report the percentage of times that a \texttt{NaN} output is returned for each number of trees. For each number of trees, we repeat the simulation 100 times, and results are shown in \autoref{fig:ranger_var}.

\begin{figure}[H]
\centering
\includegraphics[scale=.5]{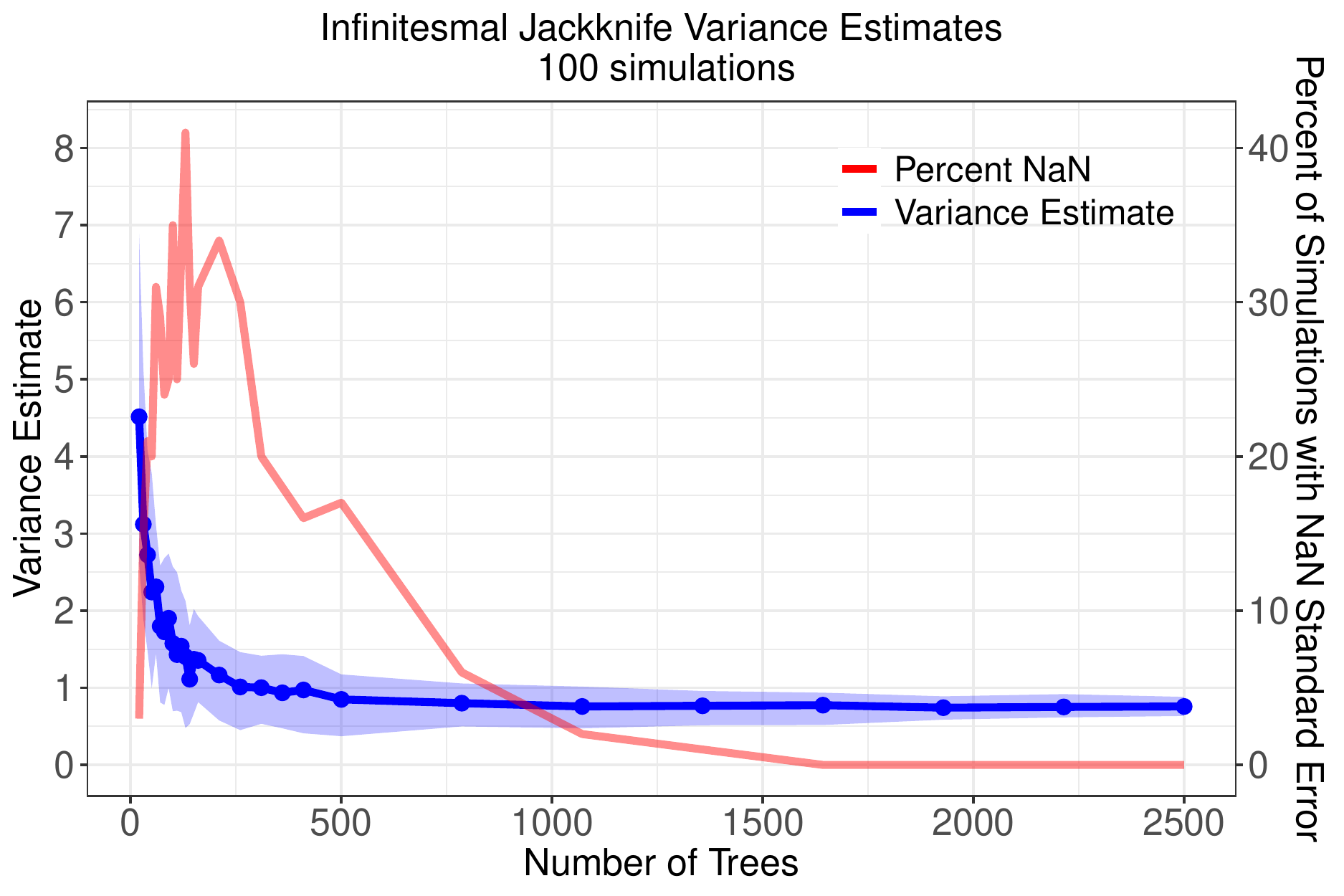}
\caption{\texttt{ranger} IJ variance estimate. Blue ribbon plot indicates central 90\% of variance estimates (corresponds to left axis), and red line (corresponds to right axis) represents percentage of runs that return \texttt{NaN}.} \label{fig:ranger_var}
\end{figure}

The IJ estimate provides overwhelmingly large variance estimates for small numbers of trees, leading to overly conservative confidence intervals and tests with exceptionally low power. Moreover, the ribbon remains quite wide until around $B = 2000$ trees, suggesting that at least $\mathcal{O}(n)$ trees are necessary to attain a stable variance estimate. A similar number of trees is necessary to ensure that a \texttt{NaN} is never returned. We should note that this is the simplest possible case of variance estimation, i.e. the estimation is only at a single point. The problem grows exponentially more complex as more test points are considered and covariance estimates are needed.  \citet{Mentch2016} note that the procedure is infeasible for more than 20-30 test points.  The authors demonstrate in follow-up work \citep{Mentch2017} that an approximate test can be produced by utilizing random projections which allows for slightly larger test sets but at the cost added computational strain.  In contrast, besides the minimal overhead required to form the additional predictions, the testing procedure proposed here is almost entirely immune to the number of points in the test set.  Once the initial predictions are formed, the only remaining work is to shuffle predictions (trees) and re-compute the difference in MSE between forests.

\subsection{Test Robustness}
We now present more figures similar to the power curves presented in \autoref{sec:sims}. The goal here is to present the proposed procedure's robustness to the number of trees $B$ and the subsample size $k_n$. To do so, we modify the simulation study plotted in the second panel of \autoref{fig:simresults}. Here, we fix the error variance at $\sigma^2(\epsilon) = 16$, and again simulate $n = 2000$ training observations and $N_t = 100$ test observations. First, we vary the number of trees built, according to
\[
B \in \big\{20, 50, 75, 125, 250, 375, 500, 750, 1000\big\}
\]
and let $k_n = n^{0.6}$. The resulting simulations are plotted in \autoref{fig:M2_by_B}.
\begin{figure}[H]
\centering
\includegraphics[scale=.6]{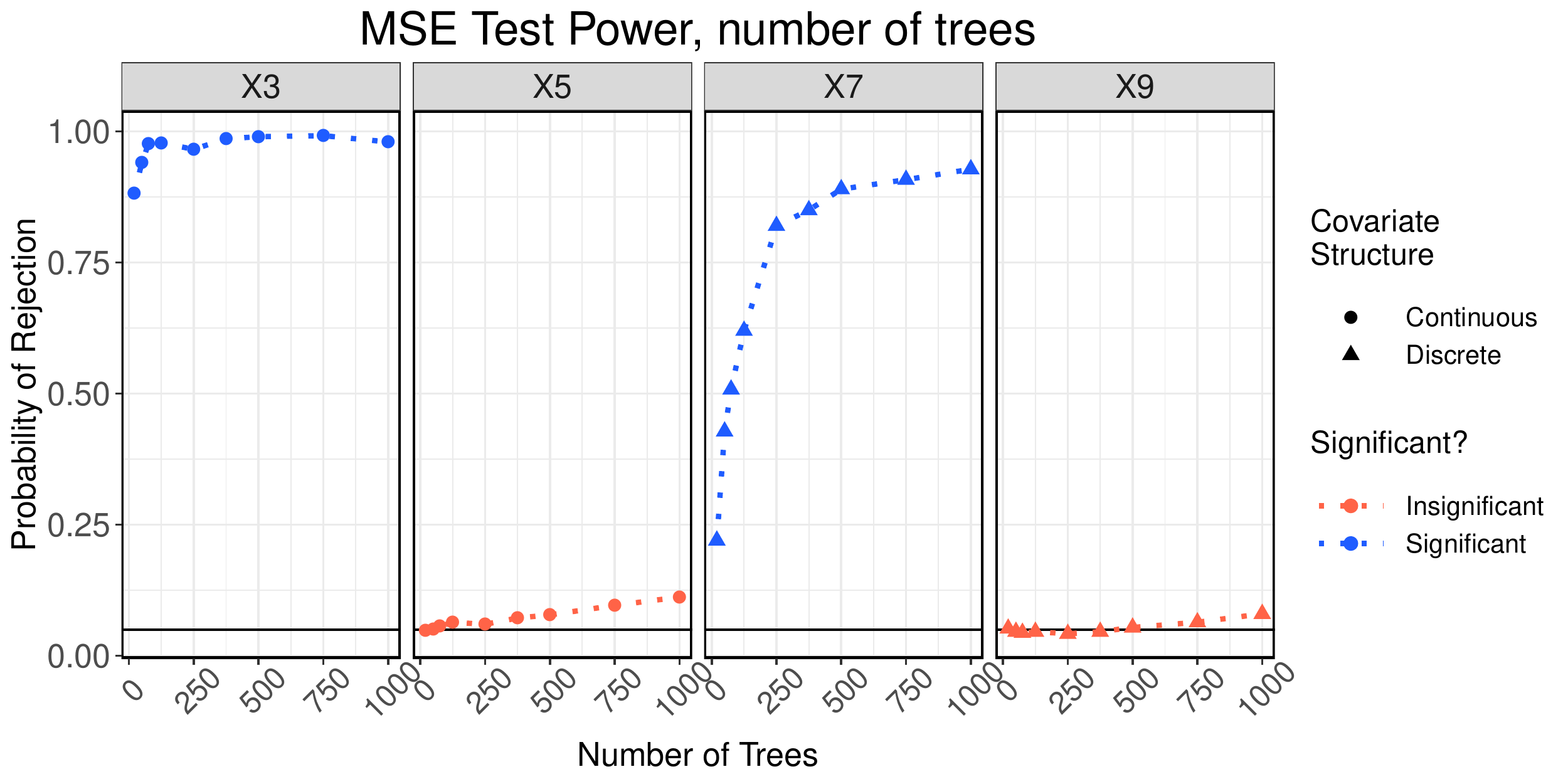}
\caption{Model 2 power curves for 500 simulations, by number of trees. The Y-axis represents $P(\tilde{p} \leq \alpha)$ where $\alpha = 0.05$ and is shown as the horizontal line across the bottom of the plots. } \label{fig:M2_by_B}
\end{figure}

Two clear patterns are clear in the figure - the power and type I error rate of the test both increase as the number of trees grows. However, the rate of growth for each of them is markedly different - the test attains high power around $B \approx 250$ trees, but deviations from the nominal level are only noticeable around $B \approx 500$ trees. Even when $B=1000$, the observed level is still within nearly 5\% of the baseline. Thus, while the level of the test may be slightly inflated for large numbers of trees, the procedure remains valid for limited, but realistic tree sizes.

Recall that the subsample size is a key limiting factor of \autoref{lem:lemma1} - it is required that $k_n = o(\sqrt{n})$ - to establish asymptotic normality. Other work \citep{Wager2018} weakens these conditions, but places explicit restrictions on the types of trees allowed in the ensemble.  We now examine the behavior of our procedure under larger sample sizes. We use the same simulation parameters as in \autoref{fig:M2_by_B}, but now fix $B = 125$ and let $k_n = n^p$, and we vary $p$ at 10 equally spaced values between 0.1 and 0.99.

The resulting simulation is shown in \autoref{fig:M2_by_p}. We see that for $p \leq 0.75$, the Type I error rate is maintained, but for larger subsamples, we begin to see a severe deviation.  Though severe, this is not necessarily surprising as such large subsampling rates correspond directly to a more severe violation of the iid approximation.

\begin{figure}[H]
\centering
\includegraphics[scale=.6]{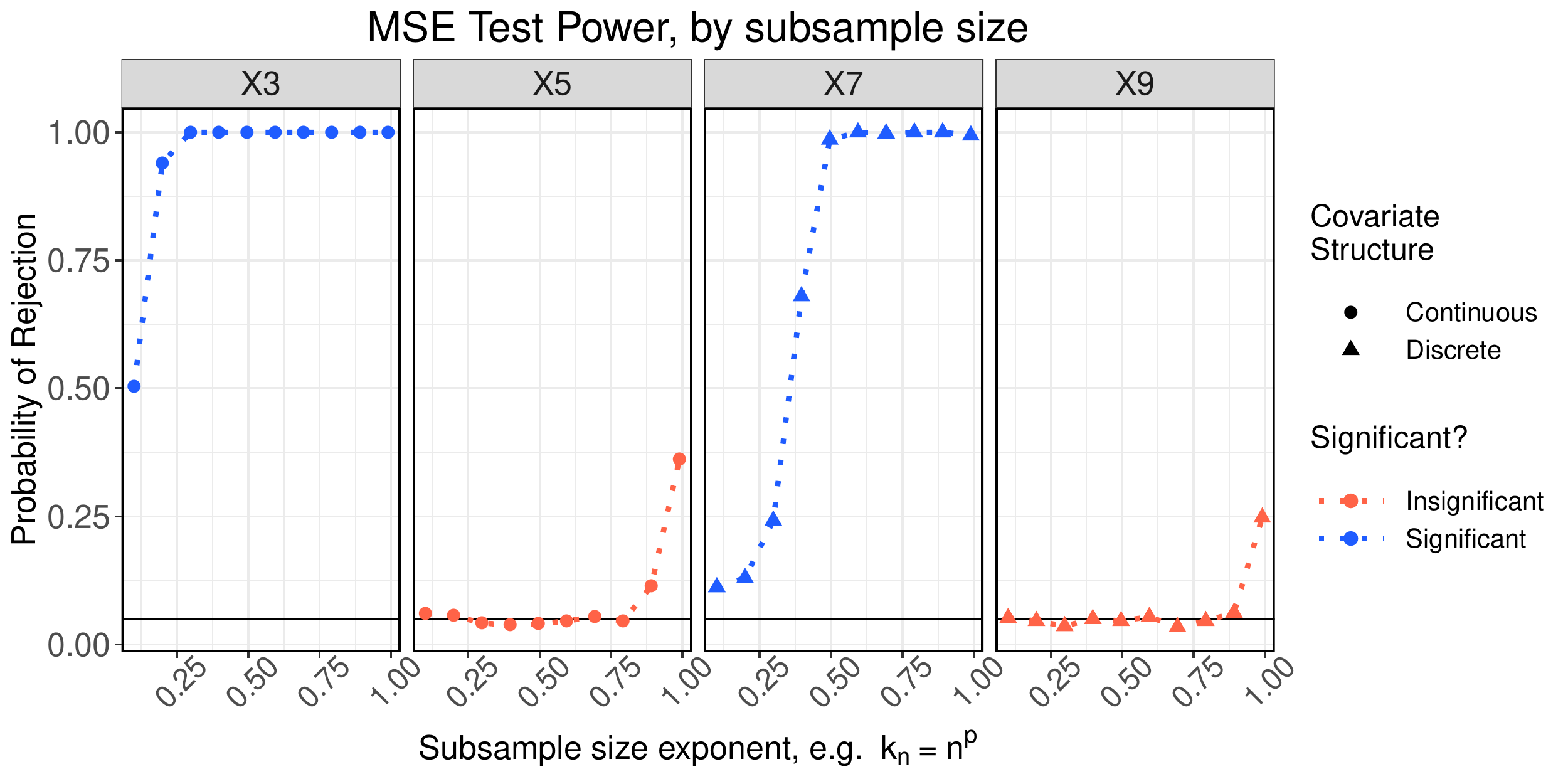}
\caption{Model 2 power curves for 500 simulations, by subsample exponent. The Y-axis represents $P(\tilde{p} \leq \alpha)$ where $\alpha = 0.05$ and is shown as the horizontal line across the bottom of the plots. } \label{fig:M2_by_p}
\end{figure}

\end{document}